\documentclass[aps,twocolumn,superscriptaddress,longbibliography,notitlepage,nofootinbib]{revtex4-2}

\usepackage{jefferymath}

\usepackage[english]{babel}
\usepackage[colorlinks=false,linkcolor=blue,citecolor=red,plainpages=false,pdfpagelabels]{hyperref}
\usepackage{lipsum} 
\usepackage{multirow}
\usepackage{makecell}
\usepackage{textcomp}
\usepackage{url}
\usepackage{tcolorbox}

\usepackage{lmodern}
\usepackage{amsfonts}
\usepackage{amsmath}
\usepackage{amssymb}
\usepackage{latexsym}
\usepackage{mathtools}
\usepackage{siunitx}

\usepackage[margin=1in]{geometry}
\usepackage{hyperref}
\hypersetup{pdfpagemode=UseNone}
\usepackage[capitalise]{cleveref}

\usepackage{amsthm}
\newtheorem{theorem}{Theorem}
\newtheorem{lemma}{Lemma}

\newtheorem{corollary}{Corollary}

\theoremstyle{definition}
\newtheorem{definition}{Definition}

\begin{document}

% \title{Quantum-Centric\! Algorithm\! for\! Sample-Based\! Krylov\! Quantum\! Diagonalization}

\title{Quantum-Centric Algorithm for Sample-Based Krylov Diagonalization}

\author{Jeffery Yu}
\email{jey@umd.edu}
\affiliation{IBM~Quantum,~IBM~T.J.~Watson~Research~Center,~Yorktown~Heights,~NY~10598,~USA}
\affiliation{Joint Center for Quantum Information and Computer Science, NIST/University of Maryland,
College Park, Maryland 20742, USA}
\affiliation{Joint Quantum Institute, NIST/University of Maryland,
College Park, Maryland 20742, USA}
\author{Javier~Robledo~Moreno}
\email{j.robledomoreno@ibm.com}
\affiliation{IBM~Quantum,~IBM~T.J.~Watson~Research~Center,~Yorktown~Heights,~NY~10598,~USA}
\author{Joseph T.~Iosue}
\affiliation{IBM~Quantum,~IBM~T.J.~Watson~Research~Center,~Yorktown~Heights,~NY~10598,~USA}
\affiliation{Joint Center for Quantum Information and Computer Science, NIST/University of Maryland,
College Park, Maryland 20742, USA}
\affiliation{Joint Quantum Institute, NIST/University of Maryland,
College Park, Maryland 20742, USA}
\author{Mirko Amico}
\affiliation{IBM~Quantum,~IBM~T.J.~Watson~Research~Center,~Yorktown~Heights,~NY~10598,~USA}
\author{Luke Bertels}
\affiliation{Quantum Information Science Section, Oak Ridge National Laboratory, Oak Ridge, TN, 37831, USA}
\author{Daniel~Claudino
}
\affiliation{Quantum Information Science Section, Oak Ridge National Laboratory, Oak Ridge, TN, 37831, USA}
\author{Bryce~Fuller}
\affiliation{IBM~Quantum,~IBM~T.J.~Watson~Research~Center,~Yorktown~Heights,~NY~10598,~USA}
\author{Peter~Groszkowski
}
\affiliation{National Center for Computational Sciences, Oak Ridge National Laboratory, Oak Ridge, TN, USA}
\author{Travis~S.~Humble
}
\affiliation{Quantum Science, Center, Oak Ridge National Laboratory, Oak Ridge, Tennessee, USA}
\author{Petar~Jurcevic}
\affiliation{IBM~Quantum,~IBM~T.J.~Watson~Research~Center,~Yorktown~Heights,~NY~10598,~USA}
\author{William Kirby}
\affiliation{IBM~Quantum,~IBM~Research~Cambridge,~Cambridge,~MA~02142,~USA}
\author{Thomas A. Maier}
\affiliation{Computational~Sciences~and~Engineering~Division,~Oak~Ridge~National~Laboratory,~Oak Ridge,~Tennessee~37831,~USA}
\author{Mario Motta}
\affiliation{IBM~Quantum,~IBM~T.J.~Watson~Research~Center,~Yorktown~Heights,~NY~10598,~USA}
\author{Bibek~Pokharel}
\affiliation{IBM~Quantum,~IBM~T.J.~Watson~Research~Center,~Yorktown~Heights,~NY~10598,~USA}
\author{Alireza Seif}
\affiliation{IBM~Quantum,~IBM~T.J.~Watson~Research~Center,~Yorktown~Heights,~NY~10598,~USA}
\author{Amir Shehata
}
\affiliation{Oak Ridge National Laboratory, Oak Ridge, Tennessee, USA}
\author{Kevin~J.~Sung}
\affiliation{IBM~Quantum,~IBM~T.J.~Watson~Research~Center,~Yorktown~Heights,~NY~10598,~USA}
\author{Minh C. Tran}
\affiliation{IBM~Quantum,~IBM~T.J.~Watson~Research~Center,~Yorktown~Heights,~NY~10598,~USA}
\author{Vinay Tripathi}
\affiliation{IBM~Quantum,~IBM~T.J.~Watson~Research~Center,~Yorktown~Heights,~NY~10598,~USA}
\author{Antonio~Mezzacapo}
\email{mezzacapo@ibm.com}
\affiliation{IBM~Quantum,~IBM~T.J.~Watson~Research~Center,~Yorktown~Heights,~NY~10598,~USA}
\author{Kunal Sharma}
\email{kunals@ibm.com}
\affiliation{IBM~Quantum,~IBM~T.J.~Watson~Research~Center,~Yorktown~Heights,~NY~10598,~USA}

\begin{abstract}
Approximating the ground state of many-body systems is a key computational bottleneck underlying important applications in physics and chemistry. 
The most widely known quantum algorithm for ground state approximation, quantum phase estimation, is out of reach of current quantum processors due to its high circuit-depths. 
Subspace-based quantum diagonalization methods offer a viable alternative for pre- and  early-fault-tolerant quantum computers.
Here, we introduce a quantum diagonalization algorithm which combines two key ideas on quantum subspaces: a classical diagonalization based on quantum samples, and subspaces constructed with quantum Krylov states.
We prove that our algorithm converges in polynomial time under the working assumptions of Krylov quantum diagonalization and sparseness of the ground state.
We then demonstrate the scalability of our approach by performing the largest ground-state quantum simulation of impurity models using a Heron quantum processors and the Frontier supercomputer. We consider both the single-impurity Anderson model with 41 bath sites, and a system with 4 impurities and 7 bath sites per impurity. Our results are in excellent agreement with Density Matrix Renormalization Group calculations. 
\end{abstract}

\maketitle

\section{Introduction}
{\let\thefootnote\relax\footnote{{This manuscript has been authored by UT-Battelle, LLC, under Contract No.~DE-AC0500OR22725 with the U.S.~Department of Energy. The United States Government retains and the publisher, by accepting the article for publication, acknowledges that the United States Government retains a non-exclusive, paid-up, irrevocable, world-wide license to publish or reproduce the published form of this manuscript, or allow others to do so, for the United States Government purposes. The Department of Energy will provide public access to these results of federally sponsored research in accordance with the DOE Public Access Plan.}}}

A significant bottleneck in physics and chemistry is the efficient estimation of the low-energy spectrum of quantum systems. Fault-tolerant quantum algorithms, including phase estimation, promise advantages over classical methods for this task, but require circuit depths beyond the reach of pre-fault-tolerant devices~\cite{kitaev1995quantum}.  
Shallow circuit approaches such as the variational quantum eigensolver~\cite{peruzzo2014variational} are more hardware-efficient, but they rely on parametric optimization and stochastic estimation of complex observables, which limits their scaling because of prohibitive runtime connected to the large number of measurements~\cite{wecker2015progress, cerezo2021variational, larocca2024review}.
This motivates the development of new quantum algorithms that can estimate the spectral properties of physical systems on current quantum computers.

Quantum diagonalization methods based on subspaces have emerged as promising algorithms for estimating spectral properties on pre-fault-tolerant devices~\cite{mcclean2017subspace,parrish2019quantum,motta2020qite_qlanczos,klymko2022realtime,Epperly_2022,shen2023realtimekrylov,yang2023dualgse,yang2023shadow,ohkura2023leveraging,kanno2023quantum,ibm2024chemistry,yoshioka2024diagonalization,motta2023subspace,oumarou2025molecular}. 
% So far, these techniques have extended the reach of quantum experiments beyond what was feasible with variational methods. 
Notably, an experimental implementation of Krylov quantum diagonalization (KQD) was shown on quantum many-body systems of up to 56 spins~\cite{yoshioka2024diagonalization}. KQD involves constructing a subspace by time-evolving a reference state over various time intervals, followed by classical diagonalization of the Hamiltonian within that subspace.
A key advantage of this approach is that convergence is guaranteed when the initial state has polynomial overlap with the ground state, and it relies on simulating quantum dynamics using circuits that can be executed at sizes beyond the reach of exact classical methods~\cite{kim2023evidence,shinjo2024unveiling}.

Subspace algorithms based on individual quantum samples~\cite{kanno2023quantum,ibm2024chemistry, kaliakin2024supramolecular,barison2024ext-sqd,liepuoniute2024triplet,shajan2024SQD-DMET} approximate ground state energies by sampling from a quantum state and performing classical post-processing on noisy data and diagonalization in quantum-centric supercomputing environments~\cite{alexeev2023quantum}. Unlike KQD, these sample-based quantum diagonalizations (SQD) do not require time-evolution circuits, making them appealing for chemistry Hamiltonians with a large number of terms. These ideas have been experimentally demonstrated for molecular electronic structure up to sizes not amenable to exact diagonalization~\cite{ibm2024chemistry}. 

In this work, we introduce a new algorithm that combines key ideas from the KQD and SQD frameworks. We refer to this algorithm as \emph{sample-based Krylov quantum diagonalization}~(SKQD), which leverages the advantages of both frameworks: the convergence guarantees of Krylov methods and the noise resilience of sample-based techniques. SKQD constructs a subspace from bistrings sampled from quantum states obtained by time-evolving a reference state over multiple intervals, and then classically diagonalizes the Hamiltonian within this subsapce to  approximate the ground state. 

We prove that under the sparsity assumptions for the ground state and given a reference state with polynomial overlap, SKQD approximates the ground state energy in polynomial time. We experimentally demonstrate the scalability of SKQD by computing ground state properties of a single-impurity Anderson model (SIAM)~\cite{Anderson1961} with 41 bath sites (42 electrons in 42 orbitals), simulated using 85 qubits and up to $6 \cdot 10^3$ two-qubit gates to prepare Krylov states; and the Anderson model with four impurities and 28 bath sites (32 electrons in 32 orbitals), using 70 qubits of a Heron processor. We show excellent agreement between SQKD and Density Matrix Renormalization Group (DMRG)~\cite{white1992DMRG1, white1993DMRG, white2005DMRG} calculations, on system sizes not amenable to exact diagonalization. To the best of our knowledge, these results are the largest simulations of the ground state properties of impurity models using heterogeneous quantum and classical architectures.

% \kunal{need to update this section or competely remove for Nature Physics submission. Let's keep it for now until we have the draft ready}
% This paper is structured as follows. In \cref{msec:background}, we summarize the Krylov and sample-based quantum diagonalization approaches. We then prove a convergence result of SKQD in \cref{msec:ksqsd-proof}, showing that under the sparsity assumptions for the ground state and given a reference state with polynomial overlap, one can approximate the ground state energy in polynomial time.
% We then numerically study the convergence with respect to the number of measurements with lattice Hamiltonians, finding that SKQD can outperform KQD under a finite shot budget.
% Finally, we present experimental results obtained on quantum processors in \cref{msec:ksqsd-experiments}. We compute ground state energies of a single-impurity Anderson model (SIAM)~\cite{Anderson1961} with 41 bath sites (42 electrons in 42 orbitals), simulated using 85 qubits and up to $6 \cdot 10^3$ two-qubit gates to prepare Krylov states. We show excellent agreement between SQKD and Density Matrix Renormalization Group (DMRG)~\cite{white1992DMRG1, white1993DMRG, white2005DMRG} calculations, on system sizes not amenable to exact diagonalization. 

\section{Sample-based Krylov Quantum Diagonalization (SKQD)}\label{msec:skqd}
We are interested in approximating the ground-state energy of a Hamiltonian $H$ defined for $n$ qubits. Let $N=2^n$. 
Let $\ket{\psi_0}$ denote an initial (reference) state. Similar to KQD (see \cref{sec:methods}), we consider the following time-evolved states, also known as Krylov states:
\begin{equation}\label{meq:krylov-basis}
\ket{\psi_k} \coloneq e^{-ik H \Delta{t}} \ket{\psi_0},
\end{equation}
where $k \in \{0, 1, \dots, d-1\}$ and $\Delta{t}$ is a chosen time step. To implement SKQD, we proceed as follows (see Fig.~\ref{fig: algorithm}):  
\begin{enumerate}
    \item Prepare a reference state $\ket{\psi_0}$.
    \item For each $k \in \{0, 1, \dots, d-1\}$, prepare $M = \mathcal{O}(\poly(n))$ copies of $\ket{\psi_k}$.
    \item Measure each $\ket{\psi_k}$ in the computational basis to obtain a sequence of bitstrings $\mathcal{X}_k = \{a_{km}~|~m=0,1,\dots,M-1\}$. 
    \item Classically project $H$ onto $\operatorname{span}(B_{d, M})$ to obtain the matrix $\widehat{\mathbf{H}}$, where $
    B_{d,M}=\{\ket{a_{km}}~|~k=0,\dots,d-1; m=0,\dots,M-1\}$. 
    \item Diagonalize $\widehat{\mathbf{H}}$ classically to find the approximation to the ground state: $|\Psi\rangle = \sum_j \mathbf{\Psi}_j |b_j\rangle$, where $b_j$ is a symbol that summarizes the unique elements $a_{km}$ in $\mathcal{X}$.
\end{enumerate}

\begin{figure}[t]
\centering
\includegraphics[width=0.95\linewidth]{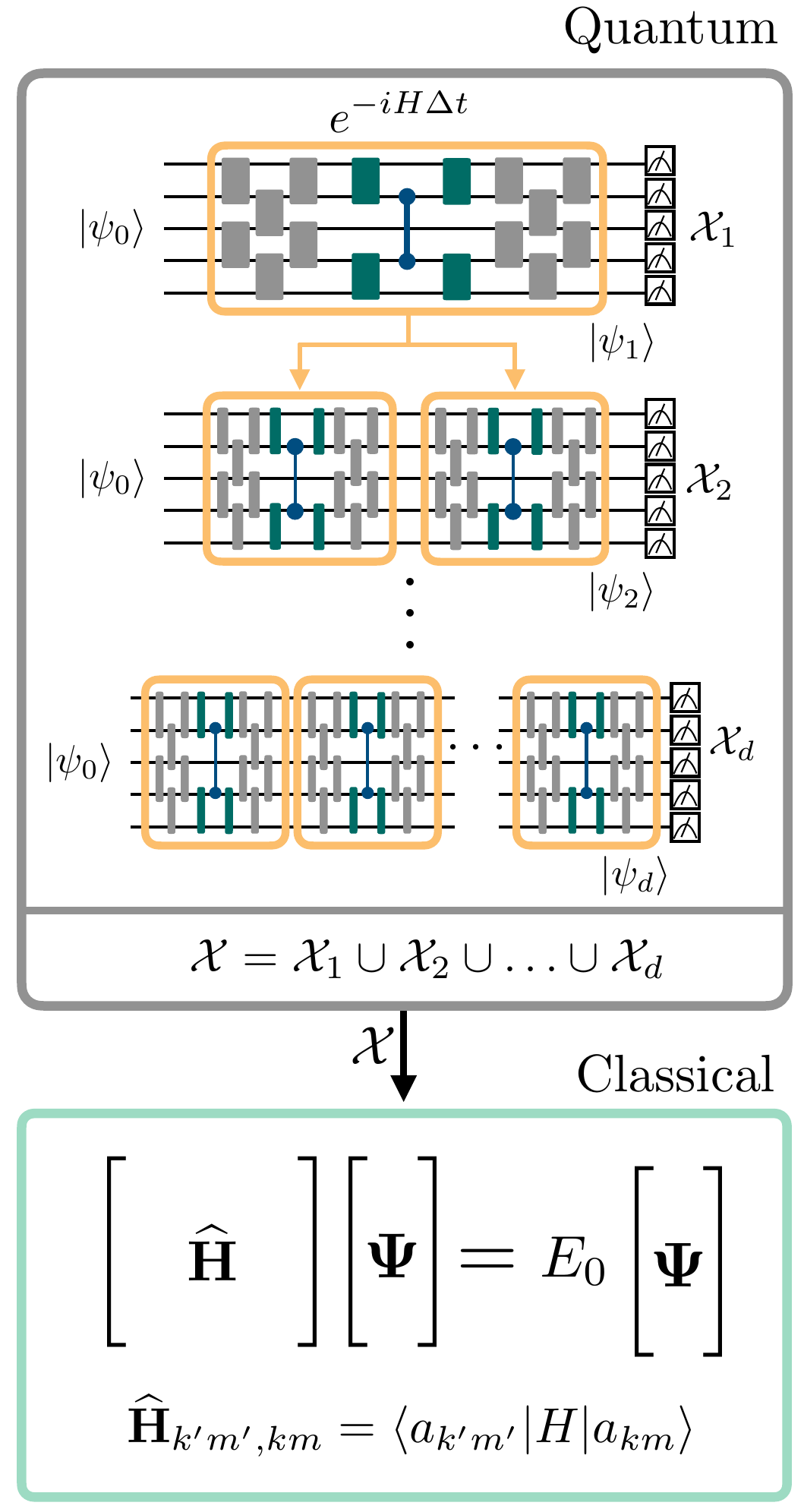}
\caption{\textbf{SQKD algorithm.} A Krylov subspace is constructed by the time evolution of a reference state $|\psi_0\rangle$ to $d$ different times. At the end of each circuit, the state is measured in the computational basis, yielding a set of measurement outcomes $\mathcal{X}$. The computational basis states sampled in $\mathcal{X}$ are used to span an approximation to the ground state of the system $|\Psi\rangle = \sum_j \mathbf{\Psi}_j |b_j\rangle$ for $b_j \in \mathcal{X}$. The components $\mathbf{\Psi}_j$ are obtained in closed form by the diagonalization of the projection of $H$ in the subspace spanned by the sampled bitstrings. }
\label{fig: algorithm}
\end{figure}

Note that the matrix $\widehat{\mathbf{H}}$ is of polynomial size and thus can be diagonalized classically, and that the $\mathbf{\Psi}_j$ components are the entries of the eigenvector of $\widehat{\mathbf{H}}$ with smallest eigenvalue. In the presence of noise, we also perform configuration recovery based on the $U(1)$ symmetry of the problem, as introduced in \cite{ibm2024chemistry}. A key feature of SKQD is its robustness to noise, which comes with additional classical cost in performing configuration recovery and diagonalization, making it well-suited for noisy quantum processors. We demonstrate this capability on IBM quantum processors in \cref{sec:experiments}. Moreover, in the Supplementary Information we show evidence of SKQD outperforming the standard KQD approach under a fixed shot budget. 

\medskip 

Before presentig our experimental results, we argue that SKQD algorithm converges efficiently under a notion of sparsity of the ground state $\ket{\phi_0}$ of $H$. We define sparsity as follows:
\begin{definition}[$(\alpha_L, \beta_L)$-sparsity]
\label{def:concentration}
For any state $\ket{\psi}$, let
\begin{equation}
    \ket{\psi} = \sum_{j=1}^Ng_j\ket{b_j},
\end{equation}
where $(b_1, \dots, b_N)$ is some ordering of length-$N$ bitstrings such that $|g_1|\geq |g_2|\geq \cdots |g_N|$.
We say that $\ket{\psi}$ exhibits \emph{$(\alpha_L, \beta_L)$-sparsity} on $\ket{b_1}$ through $\ket{b_L}$ if
\begin{equation}
    \sum_{j=1}^L |g_j|^2 \geq \alpha_L
\end{equation}
and
\begin{equation}
    |g_1|^2, \dots, |g_L|^2 \geq \beta_L.
\end{equation}
\end{definition}

Let $L$ be the smallest possible integer and $\alpha_L^{(0)}, \beta_L^{(0)}$ be the largest possible parameters such that the ground state $\ket{\phi_0}$ exhibits $(\alpha_L^{(0)}, \beta_L^{(0)})$-sparsity.

We introduce the following notation for the spectrum of $H$ and the initial state used in SKQD. Let $E_0 \le E_1 \le \dots \le E_{N-1}$ denote eigenvalues of $H$ with corresponding orthonormal eigenstates $\ket{\phi_0}, \dots, \ket{\phi_{N-1}}$. Let $\Delta E_j = E_j - E_0$ for each $0<j<N$ and let $\ket{\psi_0} = \sum_{k=0}^{N-1}\gamma_k \ket{\phi_k}$ be the eigenstate decomposition of the initial state.

\subsection{Convergence guarantees}

We argue that if the ground state exhibits sparsity as defined in \cref{def:concentration}, then SKQD can efficiently estimate the ground-state energy with bounded error. In particular, our main convergence result is as follows:
\begin{theorem}\label{mthm:skqd-thm}
Let $H$ be a Hamiltonian whose ground state $\ket{\phi_0}$ exhibits $(\alpha_L^{(0)}, \beta_L^{(0)})$-sparsity.
Let $\ket{\Psi}$ be the lowest energy state supported on the $L$ important bitstrings in $\ket{\phi_0}$.
The error in estimating the ground state energy of $H$ using SKQD is bounded by
\[
\braket{\Psi | H | \Psi} - \braket{\phi_0 | H | \phi_0} \leq \sqrt{8} \norm{H} 
\left(1 - \sqrt{\alpha_L^{(0)}}\right)^{1/2},
\]
provided all $L$ important bitstrings are sampled. The success probability of sampling all $L$ important bitstrings is at least $1 - \eta$ as long as the number of samples from each Krylov basis state exceeds $\left(d^2 \log(L/\eta)\right)/\left(|\gamma_0|^2 (\beta_L^{(0)} - 2\sqrt{\tilde{\varepsilon}})\right)$, where
\[
\tilde{\varepsilon} = 2-2\sqrt{1-\varepsilon/\Delta E_1}
\]
with 
\begin{align}\label{meq:eps-kqd}
    \varepsilon = 8 \Delta E_{N-1} \left(\frac{1 - \abs{\gamma_0}^2}{\abs{\gamma_0}^2}\right) \parens{1 + \frac{\pi \Delta E_1}{\Delta E_{N-1}}}^{-(d-1)},
\end{align}
where $|\gamma_0|^2 = \abs{\braket{\phi_0 | \psi_0}}^2$ denotes the overlap of the initial state $\ket{\psi_0}$ with the true ground state $\ket{\phi_0}$, and the timestep in \cref{meq:krylov-basis} is chosen as $\Delta t = \pi/\Delta E_{N-1}$.
\end{theorem}

Thus, our analytical bound on the additive error in approximating the ground state energy of $H$ depends on the sparsity parameter $\alpha_L^{(0)}$. Additionally, the number of samples required to find all $L$ important bitstrings is inversely proportional to $|\gamma_0|^2$, a requirement similar to the KQD method (see \cref{sec:methods}). 

To prove \cref{mthm:skqd-thm}, we develop several key results. We provide a brief proof for \cref{mthm:skqd-thm} in \cref{sec:methods} and a detailed proof in \cref{sec:kqd,app:proofs}; here we summarize the main ideas. First, we recall from \cite{Epperly_2022} that the KQD method achieves an additive error $\varepsilon$ in the energy as in \cref{meq:eps-kqd}. It then implies that the error in approximating $\ket{\phi_0}$ using $\ket{\Psi}$ is also small. We invoke the sparsity of the true ground state to show that $\ket{\Psi}$ is also sparse. Using this, we prove that each relevant bitstring in the ground state has an overlap proportional to $|\gamma_0|^2$ with at least one of the Krylov basis states. Combining these results, we argue that all $L$ relevant bitstrings can be obtained with high probability, and thus the ground state and the corresponding ground state energy can be approximated with small error.  

We note that if each time-evolution unitary $e^{-i k H \Delta t}$ incurs $\gamma$ Trotter error, then the number of samples needed from each Krylov basis state in \cref{mthm:skqd-thm} scales as $\left( \log(L/\eta)\right)/\left(|\gamma_0|^2 (\beta_L^{(0)} - 2\sqrt{\tilde{\varepsilon}})/d^2 - \gamma\right)$. We provide a proof in \cref{sec:trotter_error}.

\subsection{Experiments on impurity models}\label{sec:experiments}
To asses the accuracy of SKQD applied, we consider impurity models with increasing number of impurities as a testbed~\cite{Anderson1961,wu2022disentanglinginteractingsystemsfermionic, Barzykin1998Anderson, Holtzner2009Anderson}. 
The generic Hamiltonian for an impurity model is given by:
\begin{equation}
    H = H_{\textrm{imp.}} + H_\textrm{bath} + H_\textrm{hyb.},
\end{equation}
The impurity term describes electrons that can hop between $L$ different impurities, with a Hubbard-like onsite repulsive interaction of strength $U$:
\begin{equation}
    H_\textrm{imp.} = \sum_{\substack{l,l' = 1 \\ \sigma \in \{ \uparrow, \downarrow\}}}^L t_{ll'} \hat{d}^\dagger_{l\sigma} \hat{d}_{l'\sigma} + U \sum_{ l =1}^L \hat{d}^\dagger_{l\uparrow}\hat{d}_{l\uparrow} \hat{d}^\dagger_{l\downarrow}\hat{d}_{l\downarrow},
\end{equation}
where $\hat{d}^\dagger_{l\sigma}/\hat{d}_{k\sigma}$ are the creation/annihilation operators for impurity mode $l$ with spin $\sigma$. The symmetric matrix with $t_{ll'}$ elements describes the hopping amplitudes between impurities, and its diagonal part is a chemical potential, which is set to $t_{ll} = U/2$.

 $H_\textrm{bath}$ describes a number $K$ of non-interacting fermionic modes per impurity. Given its non-interacting nature, the bath can always be written in the single-particle basis where it is diagonal:
\begin{equation}\label{eq: bath Ham main}
    H_\textrm{bath} = \sum_{\substack{ l = 1 \\\sigma \in \{ \uparrow, \downarrow\} }}^L \sum_{\mathbf{k} = 1}^K\varepsilon_\mathbf{k} \hat{c}^\dagger_{\mathbf{k}l\sigma} \hat{c}_{\mathbf{k}l\sigma},
\end{equation}
where $\hat{c}^\dagger_{\mathbf{k}l\sigma}/\hat{c}_{\mathbf{k}l\sigma}$ are the creation/annihilation operators for mode $k$ and spin $\sigma$, associated to impurity $l$. $\varepsilon_\mathbf{k}$ is the energy of each bath mode.

\begin{figure*}[t]
\centering
\includegraphics[width=0.99\linewidth]{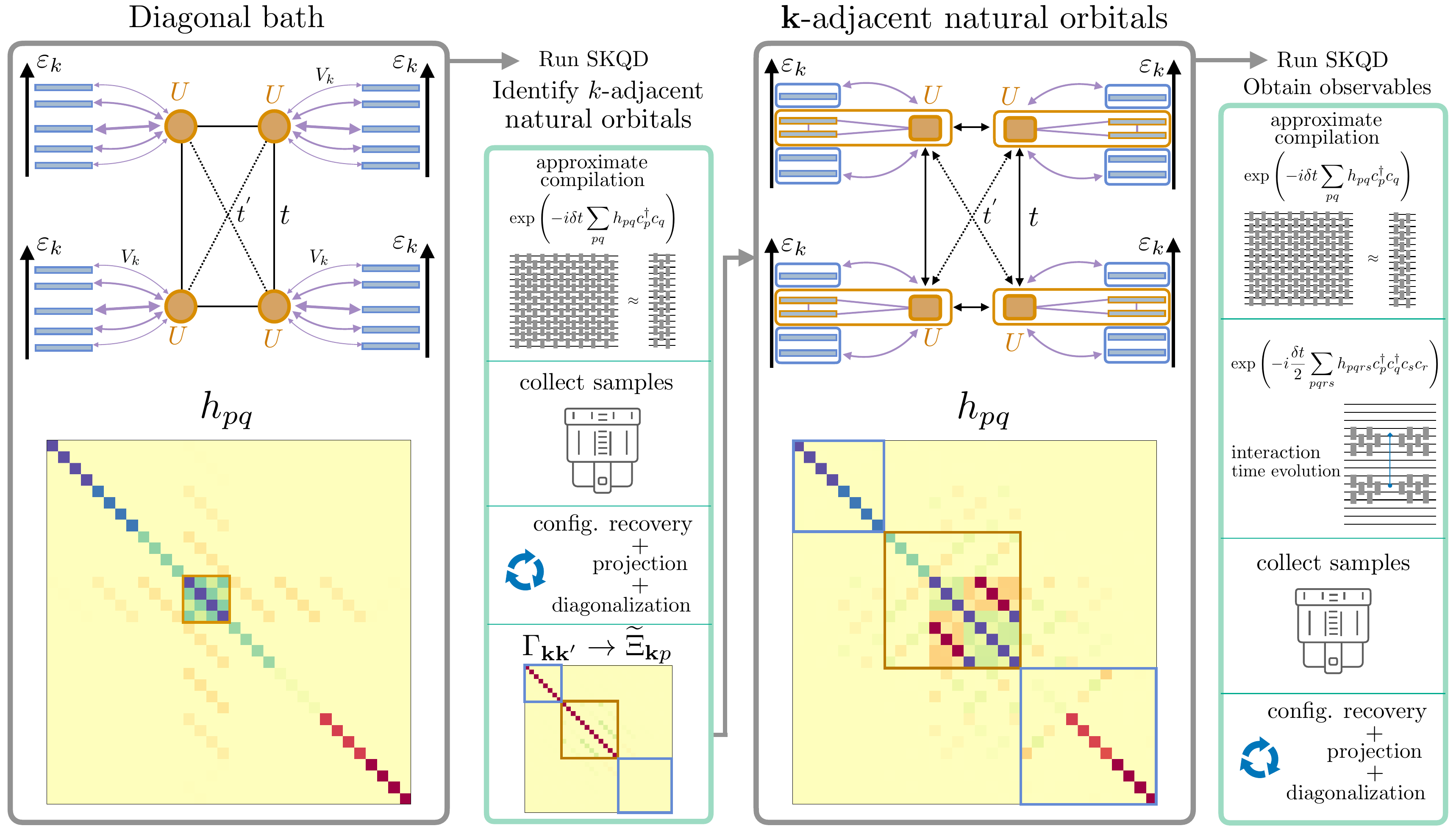}
\caption{\textbf{SQKD experimental workflow for the ground state of four-impurity model.}  From left to right: the 4-impurity model in the basis where the baths are diagonal, with the corresponding one-body matrix elements of the Hamiltonian $h_{pq}$. The brown box shows the impurity modes. SKQD is first run in this basis. The first step is the compilation of the free-fermion time evolution into a shallow circuit of Givens rotations. Then, measurement realizations are collected from the quantum device at each Trotter step, followed by an SQD ground-state estimation that uses the configuration recovery procedure, as introduced in Ref.~\cite{ibm2024chemistry}. The one-body reduced density matrix $\Gamma_{\mathbf{kk}'}$ is used to identify $\mathbf{k}$-adjacent natural orbitals, where the impurity mode is only allowed to be mixed with the bath modes corresponding to $\mathbf{k}_f$ and $\mathbf{k}_f- 1$. The resulting Hamiltonian is one where the one-body matrix elements $h_{pq}$ are close to diagonal deep in the Fermi sea and for large values of $\mathbf{k}$, and with off-diagonal two-body matrix elements. SKQD is run in this new basis, requiring the approximate compilation of the free-fermion evolution, and the efficient compilation into a constant-depth circuit of the off-diagonal two-body terms.}
\label{fig: impurity bases}
\end{figure*}

 The hybridization term describes the hopping of electrons between impurities and their corresponding bath sites:
\begin{equation}
    H_\textrm{hyb.} = \sum_{\substack{ l = 1 \\\sigma \in \{ \uparrow, \downarrow\} }}^L \sum_{\mathbf{k} = 1}^K V_\mathbf{k} \left(\hat{c}^\dagger_{\mathbf{k}l\sigma} \hat{d}_{l\sigma} + \hat{d}^\dagger_{l\sigma} \hat{c}_{\mathbf{k}l\sigma} \right)
\end{equation}
where $V_{\mathbf{k}}$ is the so-called hybridization function. Given band-width of the bath $D = \max_\mathbf{k}(\varepsilon_\mathbf{k})-\min_\mathbf{k}(\varepsilon_\mathbf{k})$, we consider semicircle-like hybridization functions $V_{\mathbf{k}} = V\sqrt{(D/2)^2 - \varepsilon_\mathbf{k}^2}$, with $V$ a parameter that controls the hybridization amplitude. We only consider the subspace corresponding to half filling and zero total magnetization.

In this manuscript we study two families of impurity models. The first is the single-impurity Anderson model (SIAM) whose $\varepsilon_\mathbf{k}$ and values of $V_\mathbf{k}$ are derived from a 1D bath geometry with open boundary conditions (see Sec.~\ref{sec:methods} for further details). We choose a system size $K = 41$ bath modes, making a total of 42 spinful modes and, values of $U = 1, 3, 7, 10$. 

The second is a 4-impurity model with $K = 7$ bath modes per impurity, making a total of 32 spinfull fermionic modes. The impurity modes are arranged in a square geometry, with hopping amplitudes  $t_{l, l+1} = t = -1$ and $t_{l, l+2} = t'= -0.5$, and a value of $U = 10$. The values of $\varepsilon_\mathbf{k}$ are sampled from a uniform distribution with $\max_\mathbf{k}(\varepsilon_\mathbf{k}) = 2$ and $\min_\mathbf{k}(\varepsilon_\mathbf{k}) = -2$ (see Sec.~\ref{sec:methods}  for the specific values). The hybridization amplitude $V = 0.16$ is chosen as the case where Heat Bath Configuration Interaction (HCI) and DMRG required the largest computational resources to converge, as detailed in Sec~\ref{sec:methods}. Figure~\ref{fig: impurity bases} shows an schematic representation of the 4-impurity model.

The similarity transformation of the Hamiltonian by a fermionic Gaussian unitary (orbital rotation, or single-particle basis change) can impact the accuracy of many-body methods in the approximation of the ground-state properties of the problem~\cite{moreno2023orbitalOptims, wu2022disentanglinginteractingsystemsfermionic, roos1980complete, head1988optimization, werner1985second, olsen2011casscf, malmqvist1990restricted, zgid_density_2008, ghosh_orbital_2008, wouters_density_2014, ma_second-order_2017}. In the limit of vanishing $V$, the ground state of the impurity models is sparse in the basis where the bath is diagonal, since the state of each bath can be described by a single Slater determinant, and the ground state of the impurities can be described by a small number of basis states. As the value of $V$ increases, it becomes energetically favorable to allow the hopping of electrons between the impurity and corresponding bath, resulting in an increased number of Slater determinants required to obtain an accurate description of the ground state. Motivated by the observation that the basis of so called natural orbitals (NOs) that diagonalizes the one-body reduced density matrix $\Gamma$ yields the set of orbitals in which the wave function is closest to a single Slater determinant, we perform a two-step SKQD experiment. The first SKQD calculation is performed in the basis where the bath is diagonal, obtaining an approximation to $\Gamma$. The reduced density matrix is diagonalized in blocks in order to find approximate NOs that only mix impurity modes with two or three bath modes for the 4-impurity and single-impurity models respectively.  The bath modes that are allowed to be mixed with the impurity mode are those closest to the Fermi level ($\varepsilon_{\mathbf{k}_f} = 0$). This procedure is summarized in Fig.~\ref{fig: impurity bases}.

We use the Jordan-Wigner~\cite{JordanWigner1928} encoding to map the fermionic degrees of freedom into the quantum processor. The second-order Trotter-Suzuki decomposition is used to realize each $|\psi_k\rangle$, as described in Sec.~\ref{sec:methods}. The mappings into the quantum processors for the SIAM and 4-impurity models are shown in panels (a) and (d) of Fig.~\ref{fig: results}.

To benchmark the accuracy of the quantum experiments, we choose DMRG as a reference classical method, since it is one of the \textit{state-of-the-art} approximate methods for single-band, single-impurity models~\cite{Holtzner2009Anderson, wu2022disentanglinginteractingsystemsfermionic, varbench}. Several physical properties are compared between the DMRG and SKQD estimations. The first is the relative error in the SKQD ground-state energy estimation, defined as $|(E_\textrm{SKQD}-E_\textrm{DMRG})/E_\textrm{DMRG}|$. Additionally, we compare the estimation of other relevant physical properties, such as two-point correlation functions.

\begin{figure*}[t]
\centering
\includegraphics[width=1\linewidth]{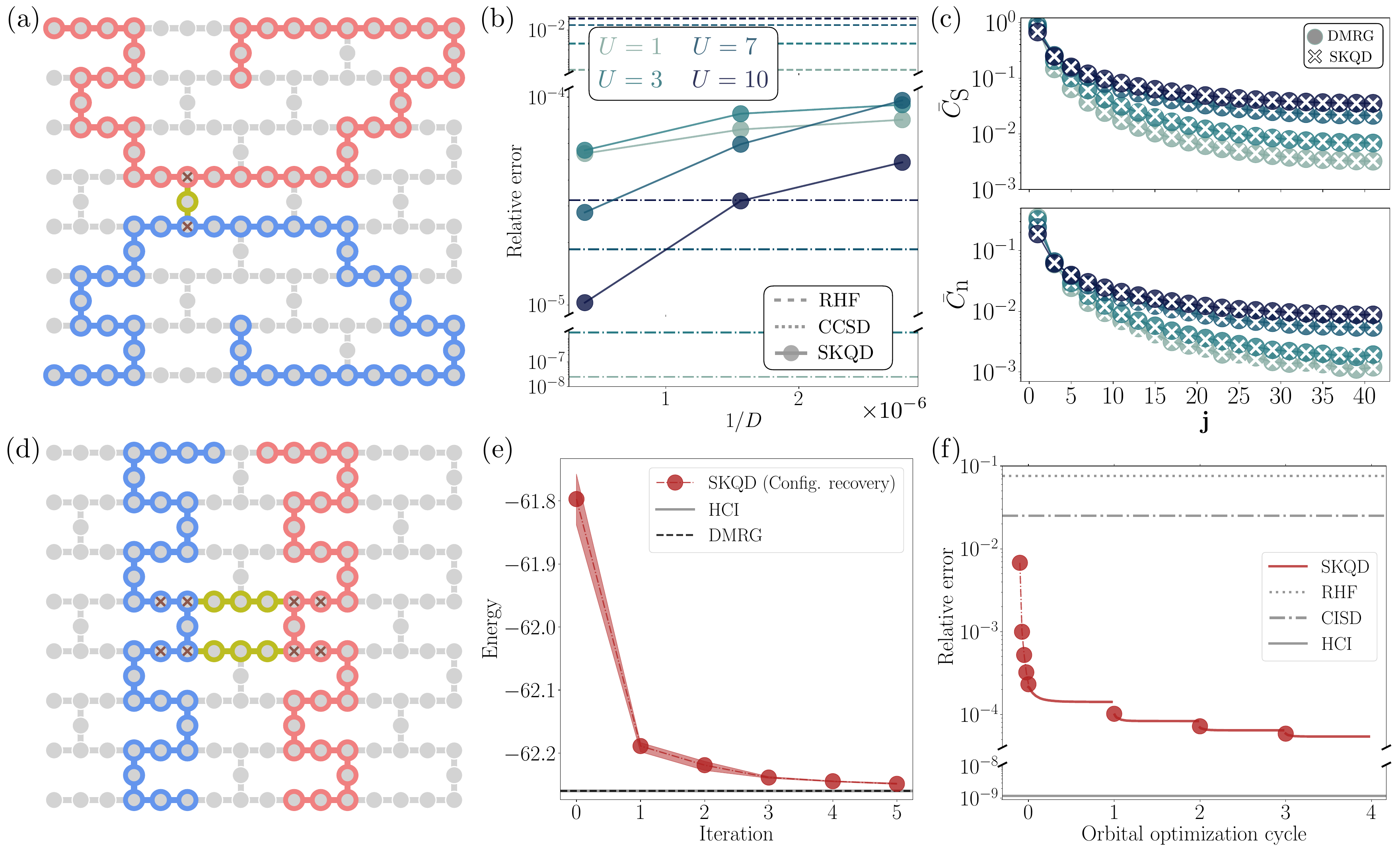}
\caption{\textbf{Experiments on quantum processors, and comparison against DMRG.} \textbf{(a)-(c)} SIAM with 41 bath sites (85-qubit experiment). 
\textbf{(a)} Shows the qubit layout, where red and blue qubits correspond to spin-up and spin-down degrees of freedom respectively. The qubit marked with a cross is the qubit representing the impurity. The green qubit is an auxiliary qubit used to implement the time evolution of the Hubbard interaction.
\textbf{(b)} Ground state energy error relative to the DMRG estimation as a function of the inverse of the SKQD subspace dimension $D$. Different colors correspond to different values of the onsite repulsion $U$ as indicated in the legend. The Hartree-Fock energies (RHF) and Coupled Cluster with Singles and Doubles (CCSD) errors are shown for reference. 
\textbf{(c)} Comparison of the two-point spin $\bar{C}_\textrm{S}$ and density $\bar{C}_\textrm{n}$ correlation functions (see Eqs.~\ref{eq: spin correlation} and~\ref{eq: density correlation}) obtained with DMRG and SKQD. 
\textbf{(d)-(f)} 4-impurity model ($L = 4$) with $K = 7$ bath models per impurity (70-qubit experiment). \textbf{(d)} Shows the qubit layout, where red and blue qubits correspond to spin-up and spin-down degrees of freedom respectively. The qubits marked with a cross correspond to the qubits representing the impurity. The green qubits are auxiliary qubits used to implement the time evolution of the Hubbard interaction. \textbf{(e)} SQKD ground state energy estimation in the $\mathbf{k}$-adjacent NOs basis as a function of the configuration recovery step. The HCI and DMRG energy estimations are shown for reference. \textbf{(f)} Effect of orbital optimizations applied to the converged SKQD ground state estimation. The error in the ground-state energy relative to the converged DMRG energy is shown as a function of the self-consistent orbital optimization cycle. The dots connected by the dashed lines correspond to the configuration recovery trajectory shown in panel (e). The solid lines show the improvement of the energy error when the orbitals are optimized to minimize the energy for fixed wave function coefficients. The dots show the error after re-diagonalizing the Hamiltonian in the new basis found by the orbital optimization procedure. The RHF, CISD and HCI errors are shown for reference.}
\label{fig: results}
\end{figure*}

Figure~\ref{fig: results} (b) shows the relative error in the SKQD ground-state energy estimation as a function of the subspace dimension on the SKQD eigenstate solver $D$. The SKQD relative error decreases from values $\sim 10^{-4}$ to $\sim 10^{-5}$ as $U$ increases from $U = 1$ to $U = 10$, which is the opposite behavior as compared to RHF and CCSD solutions. The SKQD estimations of the ground state energy become more accurate with increasing correlations in the system. This is a consequence of the increased ground-state sparsity for larger values of $U$. Panel (c) of Fig.~\ref{fig: results} compares the values of the two-point spin $\bar{C}_\textrm{S}(\mathbf{j})$ (Eq.~\ref{eq: spin correlation}) and density $\bar{C}_\textrm{n}(\mathbf{j})$ (Eq.~\ref{eq: density correlation}) correlation functions obtained from SKQD to those obtained with DMRG. The SKQD estimations are in excellent agreement with the DMRG values for all values of $\mathbf{j}$, the distance between the impurity spin and the bath spin. Additionally, in the Supplementary Materials we show that the SKQD error in the SIAM does not increase as the number of bath sites is increased from $K = 29$ to $K = 41$. 

Figure~\ref{fig: results} (e) shows the SKQD energy convergence as a function of the self-consistent configuration recovery iteration~\cite{ibm2024chemistry}, in the basis of $\mathbf{k}$-adjacent NOs. An additional retention and carry-over of the most relevant configurations is considered as compared to the procedure introduced in Ref.~\cite{ibm2024chemistry} (see Sec.~\ref{sec:methods} for additional details). The final energy estimation is on good agreement with the DMRG and HCI estimates. Once the recovery procedure has identified the relevant electronic configurations and corresponding wave-function components, we perform a final set of 5 self-consistent orbital optimizations cycles, where we alternatively optimize over single-particle basis transformations and wave function coefficients. The improvement on the energy error as a function of the orbital optimization cycle is shown in Fig.~\ref{fig: results} (f). The SKQD energy estimation has a relative error in the ground state energy of $\sim 10^{-4}$, in the basis of $\mathbf{k}$-adjacent NOs, while the relative error in the basis of optimized orbitals decreases to $\sim 6 \cdot 10^{-5}$.

\section{Discussion}
%We elaborate on our results, discussing SKQD in the context of existing quantum algorithms for ground state problems for pre-fault-tolerant and early-fault-tolerant quantum processors. 

We obtained our ground-state energy estimations via sampling from Krylov states. Consequently and in order to keep low circuit depths, the class of problems which are amenable to SKQD are those that be mapped easily using the layout of the quantum processor. Thus, \textit{ab-initio} use cases are out of reach for the SKQD implementation as presented in this manuscript in near term processors. We leave it to future investigations the adaptation of SQKD to \textit{ab-initio} problems with reduced-depth circuits. On lattice problems, SKQD resolves practically every issue of algorithms for ground states: it does not require optimization of ansatzaes, it does not incur in the quantum measurement problem, and it is robust to noisy samples since one can use configuration recovery and a classical diagonalization overhead to effectively remove the effect of noise. 

SKQD shares many convergence properties with the standard Krylov quantum diagonalization~\cite{yoshioka2024diagonalization}. However, it requires reduced circuit depths since it doesn't need to execute Hadamard tests as subroutines, and has improved noise resistance properties, as mentioned before. 
% Additionally, as we have shown, SKQD can outperform the standard KQD in terms of measurement overhead. 
While our convergence proofs requires sparsity of the ground state, in practice SKQD could  be used to study non-sparse ground states, in scenarios where basis states are captured by the Krylov circuits. 

We performed experiments with circuits up to 85 qubits and $\sim 6 \cdot 10^3$ two-qubit gates simulating the ground-state properties of the single-impurity Anderson model for different values of the onsite repulsion strength $U$, obtaining excellent agreement with DMRG and HCI calculations of the same system. Experiments were also performed in circuits with 70 qubits and $\sim 6 \cdot 10^3$ two-qubit gates studying the ground-state of generic 4-impurity models, again achieving excellent agreement with DMRG and HCI calculations. This confirms that SKQD can be used as to probe ground state physics on pre-fault-tolerant quantum computers, exceeding the reach of existing quantum methods for lattice problems, for system sizes well beyond the reach of exact diagonalization methods.

\bigskip

\textit{Note added.} While finalizing our paper, we noticed two independent papers on arXiv that share some of the ideas presented here~\cite{sugisaki2024SKQD,mikkelsen2024SKQD}.

\section{Code and data availability}
The simulation of the time evolution of fermionic Hamiltonians is carried out with the library \texttt{ffsim}~\cite{ffsim}, while configuration recovery, projection and diagonalization are carried out with the python package \texttt{qiskit-addon-sqd}~\cite{sqd_addon}. Quantum circuits are generated and transpiled using \texttt{qiskit}~\cite{qiskit2024}. DMRG calculations are performed using the \texttt{block2} package~\cite{block2}. HF and CCSD calculations are performed using the PySCF library~\cite{sun2018pyscf,sun2020recent}.

\begin{acknowledgments}
    We acknowledge helpful discussions with 
    Jay Gambetta,
    Toshinari Itoko,
    and Caleb Johnson. We acknowledge the use of tacokit and are grateful to Zlatko Minev for development of the package and troubleshooting support. This research was supported by the Quantum Science Center, a National Quantum Science Initiative of the Department of Energy (DOE), managed by the Oak Ridge National Laboratory (ORNL).  The work by T.A.M. (model selection and analysis) was supported by the U.S. Department of Energy, Office of Science, Basic Energy Sciences, Materials Sciences and Engineering Division. D.C. and L.B. acknowledge support from the Laboratory Directed Research and Development Program of Oak Ridge National Laboratory, managed by UT-Battelle, LLC, for the US Department of Energy. This material is based upon work supported by the U.S. Department of Energy, Office of Science, National Quantum Information Science Research Centers, Quantum Science Center (QSC). This research used resources of the Oak Ridge Leadership Computing Facility, which is a DOE Office of Science User Facility supported under Contract DE-AC05-00OR22725.
\end{acknowledgments}

\section*{Author Contributions}
Design of the workflow and experiments: J.Y., J.R.-M., P.G., W.K., T.A.-M., M.M., K.J.-S., M.T., A.M., and K.S. Implementation of the workflow: J.R.-M., M.A., P.G., P.J., M.M., B.P., and V.T.  Execution of the quantum part of the workflow on Heron: B.P., P.J., and V.T. Execution of the classical part of the workflow on Frontier: P.G. Numerical benchmarks: J.Y., J.R.-M., J.T.-I., M.M., and K.J.S. Analytical derivations: J.Y., J.R.-M., J.T.-I.,  W.K., M.T., and K.S. All authors contributed to the manuscript writing and data analysis.

\section{Methods}\label{sec:methods}

Before providing proof details of \cref{mthm:skqd-thm}, we recall prior results on KQD and SQD. 
As discussed in \cref{msec:skqd}, let $\ket{\psi_0}$ denote an initial (reference) state. In KQD, the quantum Krylov subspace $\mathcal{S}$ is generated by the time-evolved states $\ket{\psi_k}$ as in \cref{meq:krylov-basis}, 
where $k \in \{0, 1, \dots, d-1\}$. This construction reduces the exponentially large $N$-dimensional Hilbert space to a subspace of dimension~$d$. By projecting the Hamiltonian onto $\mathcal{S}$, we get the following generalized eigenvalue problem:
\begin{equation}
\label{eq:gen-eig}
\mathbf{\widetilde{H}} v = \widetilde{E} \mathbf{\widetilde{S}} v,
\end{equation}
where $v$ is a coordinate vector corresponding to the state ${\ket{\widetilde\phi}=\sum_{k=0}^{d-1} v_k\ket{\psi_k} \in \mathcal{S}}$, and 
\begin{equation}
\label{eq:gen-eig-mats}
\mathbf{\widetilde{H}}_{jk} \coloneq \braket{\psi_j | H | \psi_k}, \qquad \mathbf{\widetilde{S}}_{jk} \coloneq \braket{\psi_j | \psi_k}.
\end{equation}

% Let $E_0 \le E_1 \le \dots \le E_{N-1}$ denote eigenvalues of $H$ with corresponding orthonormal eigenstates $\ket{\phi_0}, \dots, \ket{\phi_{N-1}}$. Let $\Delta E_j = E_j - E_0$ for each $0<j<N$ and let $\ket{\psi_0} = \sum_{k=0}^{N-1}\gamma_k \ket{\phi_k}$ be the eigenstate decomposition of the initial state. 
Let $\widetilde{E}$ denote the ground-state energy approximation obtained from the Krylov quantum subspace approach. From the results of Epperly \emph{et al.} in~\cite{Epperly_2022}, it follows that in the ideal, noise-free case, $0\leq \widetilde{E}-E_0 \leq \varepsilon$, as in \cref{meq:eps-kqd}, where $E_0$ is the true ground-state energy. 

From \cref{meq:eps-kqd}, it follows that if the initial state has a nontrivial ground-state overlap $\abs{\gamma_0}^2 = \Theta(1)$, and if $H$ has a well-behaved spectrum (i.e., $\Delta E_{N-1}$ not growing too quickly and $\Delta E_1$ not too small), then a constant error $\varepsilon$ in approximating the ground-state energy can be achieved by setting $d = O(\log(1/\varepsilon))$ (see~\cref{sec:kqd} for a review of this analysis).  
Moreover, for an analysis of noisy KQD, we refer readers to~\cite{Epperly_2022,kirby2024analysis}. 

We now summarize the SQD framework for approximating ground-state energies, analyzing its behavior against a notion of sparsity of the ground state $\ket{\phi_0}$ of $H$. Similar to \cref{msec:skqd}, we consider $L$ to be the smallest possible integer and $\alpha_L^{(0)}, \beta_L^{(0)}$ be the largest possible parameters such that the ground state $\ket{\phi_0}$ exhibits $(\alpha_L^{(0)}, \beta_L^{(0)})$-sparsity.
If \( L \) bitstrings \( \ket{b_1}, \dots, \ket{b_L} \) can be sampled with high probability from a quantum circuit \cite{ibm2024chemistry}, then \(H\) can be represented in the subspace spanned by \( \{\ket{b_j}\}_{j=1}^L \), yielding a projected matrix \(\widehat{\mathbf{H}}\). By classically diagonalizing \(\widehat{\mathbf{H}}\), we obtain an approximation of \(E_0\). In \cite{ibm2024chemistry}, it was shown that using this approach, one can approximate \(E_0\) with an additive error of up to \(\varepsilon' = 2\sqrt{2}\Vert H \Vert \left(1-\sqrt{\alpha_L^{0}}\right)^{1/2}\) and a success probability of at least \(1-\eta\), provided that the number of samples obtained from the state exceeds \(2/\beta^{(0)}_L \log(1/\eta)\).

The KQD approach can achieve good accuracy in estimating the ground-state energy, provided the initial state has a nontrivial overlap with $\ket{\phi_0}$ and the spectrum of $H$ is well-behaved. However, this approach requires estimating the matrix elements of $\widetilde{\mathbf{H}}$ and $\mathbf{\widetilde{S}}$, as defined in \cref{eq:gen-eig-mats}, necessitating $O(1/\epsilon^2)$ samples to achieve error $\epsilon$, 
and additional sample overhead for noise mitigation. While these challenges do not rule out near-term implementation, as shown by~\cite{yoshioka2024diagonalization}, finding ways to circumvent them is clearly desirable. On the other hand, the SQD approach is more natural for near-term devices and well-suited for problems where the target ground state energy can be captured by a sparse wavefunction.

\subsection{Convergence proof}

We prove the convergence as follows. We begin with the sparsity assumption for ground states. Let the ground state $\ket{\phi_0}$ of $H$ exhibits $(\alpha_L^{(0)}, \beta_L^{(0)})$-sparsity, in the sense of \cref{def:concentration}. 
For our convergence proof, we assume the ground state is sparse in the sense that $L = \mathcal{O}(\poly(n))$. We provide a comparison of our sparsity assumption with those proposed in \cite{bravyi2024classical, aaronson2024verifiable} in \cref{sec:app-sparsity-notions}. 
% To implement our approach, we proceed as follows: for each $k \in \{0, 1, \dots, d-1\}$ in \cref{meq:krylov-basis}, we prepare $M = \mathcal{O}(\poly(n))$ copies of the state $\ket{\psi_k}$ and measure each in the computational basis to obtain a sequence of bitstrings $\{a_{km}~|~m=0,1,\dots,M-1\}$. We then classically project $H$ into the subspace spanned by
% $
%     B_{d,M}=\{\ket{a_{km}}~|~k=0,\dots,d-1; m=0,\dots,M-1\}$, and denote the resulting matrix as $\widehat{\mathbf{H}}$. Finally, we solve the eigenvalue problem for $\widehat{\mathbf{H}}$, which is of polynomial size and thus can be computed classically. 

We first prove that if the KQD method achieves an additive error $\varepsilon$ in the energy as in \cref{meq:eps-kqd}, then the error in the corresponding approximate ground state $\ket{\widetilde\phi}$ is
\begin{equation}
\label{eq:state_diff}
    {\norm{\ket{\widetilde\phi} - \ket{\phi_0}}^2 \leq \tilde{\varepsilon} = O\left(\frac{\varepsilon}{\Delta E_1}\right)}.
\end{equation}
Next, we show that if the true ground state $\ket{\phi_0}$ exhibits $(\alpha_L^{(0)}, \beta_L^{(0)})$-sparsity, then $\ket{\widetilde\phi}$ also exhibits $(\alpha_L, \beta_L)$-sparsity, with
\begin{equation}
\label{eq:alpha_L_beta_L}
    \alpha_L = \alpha_L^{(0)} - 2 \sqrt{\tilde{\varepsilon} },\quad \beta_L = \beta_L^{(0)} - 2 \sqrt{\tilde{\varepsilon} }.
\end{equation}

We denote the set of important bitstrings defining the ground state as $B = \{|b_i\rangle\}_{i=1}^L$, where a bitstring’s importance is determined by $|g_i|$, as defined in \cref{def:concentration}. We prove that each bitstring in $B$ has overlap proportional to $|\gamma_0|^2$ with at least one of the Krylov basis states.
We express each $\ket{\psi_k}$ in the computational basis as $
\ket{\psi_k} = \sum_{j=1}^N c^{(k)}_j \ket{b_j}$.
Assume that ${\ket{\widetilde\phi} = \sum_{k=0}^{d-1} d_k \ket{\psi_k}}$, where $\ket{\psi_k}$ is defined in \cref{meq:krylov-basis} and $|d_k| \leq \frac{1}{|\gamma_0|}$ for each $k=0,1,...,d-1$ (note that the latter condition is nontrivial because the $\ket{\psi_k}$ are nonorthogonal, so the squared norms of the $d_k$ need not sum to unity).
We then show that for each $1 \leq j \leq L$, there exists some $k$ such that
\begin{equation}
    |c^{(k)}_j|^2 \geq \frac{|\gamma_0|^2 \beta_L}{d^2},
\end{equation}
where one should note that $\beta_L$ also depends on $d$ via \cref{eq:state_diff,eq:alpha_L_beta_L} and the fact that the error $\varepsilon$ from KQD converges with increasing $d$.
Thus, one can efficiently obtain the $L$ important bitstrings by sampling from Krylov basis states, provided the initial state $\ket{\psi_0}$ has overlap $|\gamma_0|^2 \in \mathcal{O}($1/\text{poly}(n))$ $ with the exact ground state.

\subsection{Impurity model parameters}

\paragraph{Single-impurity Anderson model.--}
In this case the impurity Hamiltonian is simplified to:
\begin{equation}
    H_\textrm{imp} = \frac{U}{2} \left( \hat{n}_{d\uparrow} + \hat{n}_{d\downarrow} \right) + U \hat{n}_{d\uparrow}\hat{n}_{d\downarrow}.
\end{equation}
We consider a bath corresponding to a 1D chain of length $K$ with open boundary conditions, sharing the topology of the impurity model considered in Ref.~\cite{wu2022disentanglinginteractingsystemsfermionic}:
\begin{equation}\label{eq: bath Ham}
    H_\textrm{bath} = -t \sum_{\substack{\mathbf{j} = 0\\ \sigma\in \{\uparrow, \downarrow\}}}^{K-1} \left(\hat{c}^\dagger_{\mathbf{j}, \sigma}\hat{c}_{\mathbf{j}+1, \sigma} + \hat{c}^\dagger_{\mathbf{j}+1, \sigma}\hat{c}_{\mathbf{j}, \sigma} \right),
\end{equation}
Where the symbol $\mathbf{j}$ labels the bath sites in the position basis. The hybridization term describes the hopping between the first bath site and the impurity:
\begin{equation}
    H_\textrm{hyb} = V\sum_{\sigma \in \{\uparrow, \downarrow \}} \left(\hat{d}^\dagger_\sigma \hat{c}_{0, \sigma} + \hat{c}^\dagger_{0, \sigma} \hat{d}_{\sigma} \right).
\end{equation}
Due to the approximate translational symmetry present in the bath, we do not expect the ground state to be sparse in the position basis (Eq.~\ref{eq: bath Ham}). Since the bath is non-interacting, there exists a single-particle basis transformation that diagonalizes $H_\textrm{bath}$. We refer to the transformed basis as the basis where the bath is diagonal, whose fermionic modes are labeled by the sub-index $\mathbf{k}$ :
\begin{equation}
    \hat{c}^\dagger_{\mathbf{k}, \sigma} = \sum_{\mathbf{j} = 0}^{K-1} \Xi_{\mathbf{jk}} \hat{c}^\dagger_{\mathbf{j}, \sigma},
\end{equation}
where $\mathbf{k} = 0, \hdots, K-1$ and $\Xi\in \mathbb{R}^{L \times L}$ the orthonormal matrix that diagonalizes the $L \times L$ hopping matrix $T$:
\begin{equation}
    T = 
    \begin{bmatrix}
        0 & -t & 0 & 0 \hdots & 0 \\
        -t & 0 & -t & 0 \hdots & 0 \\
        0 & \ddots & \ddots & \ddots & 0 \\
        0 & \hdots & -t & 0 & -t\\
        0 & \hdots &  0& -t & 0\\
    \end{bmatrix}.
\end{equation}
While this is close to a momentum basis, each column in $\Xi$ does not exactly correspond to a basis vector of the discrete Fourier transform in a one-dimensional domain with $L$ points, due to the choice of open boundary conditions. In the basis that diagonalizes the bath, the bath and hybridization Hamiltonians take the form:
\begin{equation}
    H_\textrm{bath} = \sum_{\substack{\mathbf{k} = 0 \\ \sigma \in \{\uparrow, \downarrow\}}}^{L-1} \varepsilon_\mathbf{k} \hat{n}_{\mathbf{k}, \sigma},
\end{equation}
with $\varepsilon_\mathbf{k}$ the eigenvalues of $T$, and:
\begin{equation}
    H_\textrm{hyb} =  \sum_{\substack{\mathbf{k} = 0 \\ \sigma \in \{\uparrow, \downarrow\}}}^{L-1} V_{\mathbf{k}} \left(\hat{d}^\dagger_\sigma \hat{c}_{\mathbf{k}, \sigma} + \hat{c}^\dagger_{\mathbf{k}, \sigma} \hat{d}_{\sigma} \right).
\end{equation}
With $V_\mathbf{k} = V \cdot \Xi_{0\mathbf{k}} $. Note that the locality of the hybridization term is lost in favor of a sparse representation of the ground state. 

In this model the two-point spin and density correlation functions are relevant to study phenomena like the Kondo screening length~\cite{Barzykin1998Anderson, Borda2007Anderson, Holtzner2009Anderson}, or their universal collapse for different values of $U$~\cite{Holtzner2009Anderson, wu2022disentanglinginteractingsystemsfermionic}.
The staggered spin-spin correlation function is defined as:
\begin{equation}\label{eq: spin correlation}
    \bar{C}_\textrm{S}(\mathbf{j}) = (-1)^\mathbf{j} \bigg[ \left\langle \hat{\vec{S}}_d \cdot \hat{\vec{S}}_\mathbf{j} \right\rangle - \left\langle \hat{\vec{S}}_d \right\rangle \cdot \left\langle \hat{\vec{S}}_\mathbf{j} \right\rangle \bigg],
\end{equation}
where the spin operators are defined as: $\hat{S}^\mu_\mathbf{j} = \sum_{\alpha \beta} \sigma^\mu_{\alpha \beta} \hat{c}^\dagger_{\mathbf{j}, \alpha} \hat{c}_{\mathbf{j}, \beta}$, with $\mu = x, y, z$, and $\hat{S}^\mu_d = \sum_{\alpha \beta} \sigma^\mu_{\alpha \beta} \hat{d}^\dagger_{\alpha} \hat{d}_{\beta}$. The second one is the staggered density-density correlation function:
\begin{equation}\label{eq: density correlation}
    \bar{C}_\textrm{n}(\mathbf{j}) = (-1)^\mathbf{j} \sum_{\sigma \in \{\uparrow, \downarrow \}}\bigg[ \left\langle \hat{n}_{d\sigma}  \hat{n}_{\mathbf{j}\sigma} \right\rangle - \left\langle \hat{n}_{d\sigma} \right\rangle  \left\langle  \hat{n}_{\mathbf{j}\sigma} \right\rangle \bigg].
\end{equation}

\paragraph{4-impurity model.--}
The values of $\varepsilon_{\mathbf{k}}$ are shown in Fig.~\ref{fig:dispersion relation}, corresponding to the values obtained from drawing 7 random real numbers with uniform probability in the range $[-1, 1]$, and translated and scaled such that $\min_{\mathbf{k}}\varepsilon_{\mathbf{k}} = -2$ and $\max_{\mathbf{k}}\varepsilon_{\mathbf{k}} = 2$. The choice of a random dispersion relation for the bath modes is motivated by the desire to eliminate as much as possible the effect of an specific underlying lattice geometry for the bath.

\begin{figure}[h!]
\centering
\includegraphics[width=0.5\textwidth]{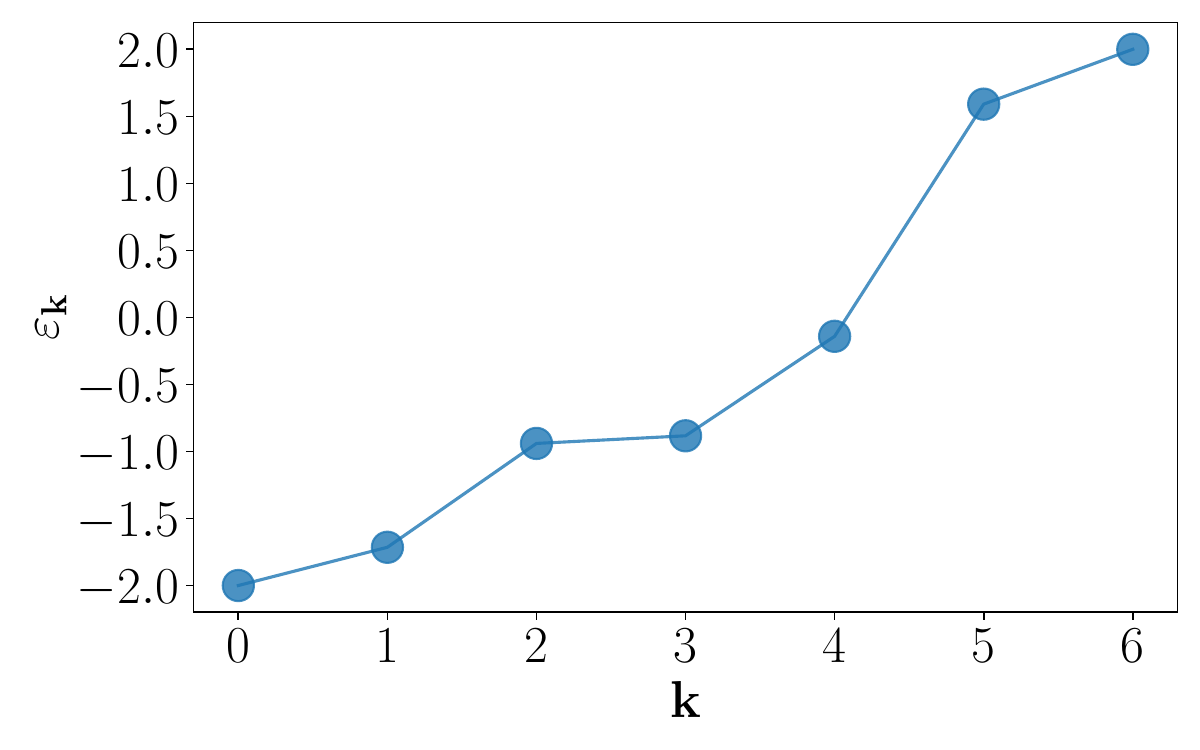}
\caption{Dispersion relation of the bath for the 4-impurity model. See Equation~\ref{eq: bath Ham main}.}
\label{fig:dispersion relation}
\end{figure}

\subsection{Computational complexity of the approximation of the ground state of the 4-impurity model as a function of $V$}

One of the goals of this study was to understand the performance of SKQD in relation to typical electronic structure methods for classical computers, as well as to identify the more challenging regimens for the classical methods, and focus on those regimes in the use of SQKD.
To this end, we focus on the 4-impurity model and the following values of the hybridization parameter, $V = 0.16, 0.32, 0.40, 0.60, 0.80, 1.60, 2.00$. For these Hamiltonians, we ran standard electronic methods: restricted closed-shell Hartree-Fock (RHF), Moller-Plesset second-order perturbation theory (MP2), configuration interaction singles and doubles (CISD), coupled-cluster singles and doubles (CCSD), and CCSD with perturbative triples (CCSD(T)) as implemented in PySCF~\cite{sun2018pyscf,sun2020recent} (we remark that CCSD does not converge for $V = 0.16$, and here we report the solution leading to the smallest (T) correction). In addition, we ran the density-matrix renormalization group (DMRG) as implemented in Block2~\cite{zhai2023block2}, and the heat-bath configuration interaction (HCI) as implemented in DICE~\cite{sharma2017semistochastic}. \\
We ran HCI using the basis of CCSD natural orbitals, using truncation thresholds of $\varepsilon = x \cdot 10^{-y}$ with $x=5,1$ and $y=3,4,5,6,7$.
The choice of natural orbitals is compelling for HCI, since natural orbitals may lead to a more compact representation of the ground-state wavefunction as a linear combination of Slater determinants than, e.g., molecular orbitals.
We ran DMRG using the basis of CCSD natural orbitals (ordered with a genetic algorithm), enforcing $SU(2)$ symmetry, and using a schedule with bond dimensions $D = 100,200,400,600,800,1000,2000,3000,4000$ and 8 sweeps per bond dimension, which allows us to converge the DMRG energy well within $10^{-10}$ in units of energy.

\begin{figure}[h!]
\includegraphics[width=0.45\textwidth]{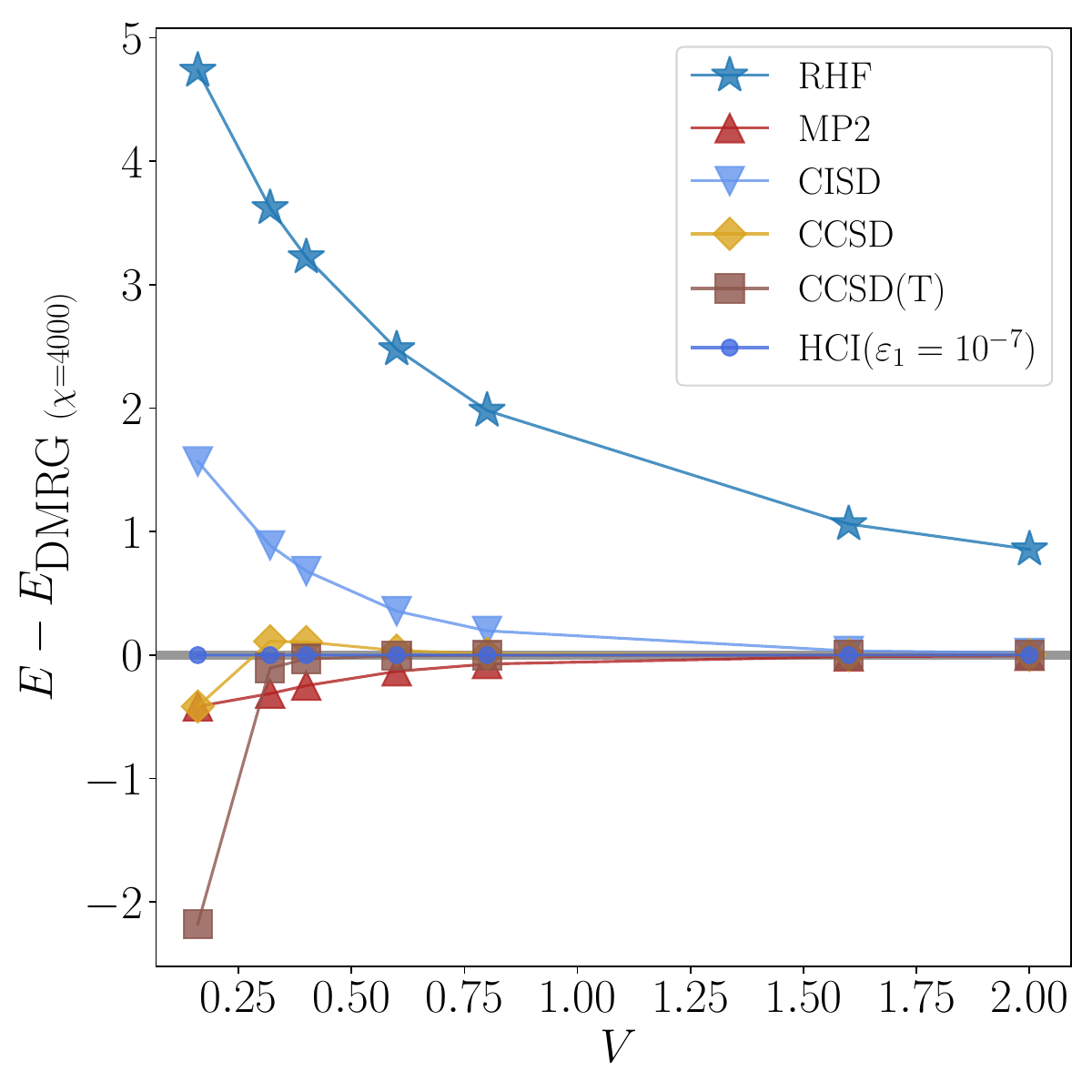}
\caption{Ground-state energy of the four-impurity model using the classical methods RHF, MP2, CISD, CCSD, and CCSD(T), HCI with truncation threshold $10^{-7}$ , and DMRG with bond dimension $D=4000$.}
\label{fig:classical_methods}
\end{figure}

The total energy from various electronic structure methods is shown in Fig.~\ref{fig:classical_methods}. RHF and CISD overestimate the ground-state energy, and MP2 underestimates it, to an extent that increases with decreasing $V$. On the other hand, where CCSD and CCSD(T) converge ($V > 0.16$), these methods agree with each other, with HCI, and with DMRG well within 0.001 in energy units. In fact, the largest deviation between the HCI and DMRG energies reported in the figure is $\sim  10^{-7}$ energy units.

In Fig.~\ref{fig:accuracy_scaling} we explore how the energy error of HCI and DMRG depends on the computational cost, defined as 
\begin{equation}
\varepsilon_{\mathrm{method}} = E_{\mathrm{method}}(P) - E_{\mathrm{ref}}
\end{equation}
where $E_{\mathrm{ref}}$ is a reference value for the ground-state energy (here, the DMRG energy with bond dimension 4000) and $P$ is a parameter that controls the computational cost of the method, i.e. the number $d_{\mathrm{HCI}}$ of configurations in HCI and the bond dimension $\chi$ of the MPS state in DMRG.

For both methods, the energy error follows a power law dependence on the inverse of the cost control parameter $1/P$ (either the subspace dimension for HCI or the bond dimension for DMRG), with exponents and prefactors that depend on both $V$ and the method. In other words,
\begin{equation}
\begin{aligned}
\varepsilon_{\mathrm{DMRG}}(V) & = C_{\mathrm{DMRG}}(V) \left( \frac{1}{D} \right)^{m_{\mathrm{DMRG}}(V)} \\
\varepsilon_{\mathrm{HCI}}(V) & = C_{\mathrm{HCI}}(V) \left( \frac{1}{d_{\mathrm{HCI}}} \right)^{m_{\mathrm{HCI}}(V)} \;.
\end{aligned}
\end{equation}
Such a power-law scaling, combined with the empirical observation that HCI and DMRG do not appear to be converging on local minima of the energy, means that decreasing the error of HCI and DMRG calculations below a desired threshold $\delta$ using a classical computer does not require an exponential cost in $1/\delta$, either in memory or in runtime. Therefore, the only way to achieve competitive results for SKQD or any other method with systematically improvable accuracy is to achieve a power-law behavior with a lower exponent and/or prefactor.

The results shown in Figs.~\ref{fig:classical_methods} and~\ref{fig:accuracy_scaling} show that the problem becomes harder for DMRG and HCI as the value of $V$ is decreased, since the system becomes less dominated by the kinetic energy terms. We choose the lowest value of $V = 0.16$ to test the accuracy of SKQD on hardware experiments.

%For the representative case of $\eta = 1.32$, we have $m_{\mathrm{HCI}} = 1.37$ and $C_{\mathrm{HCI}} = e^{11.07}$, while $m_{\mathrm{DMRG}} = 4.35$ and $C_{\mathrm{DMRG}} = e^{7.87}$. As a consequence, to achieve the same energy error $\varepsilon$, HCI and DMRG need to employ a $d_{\mathrm{HCI}}$ and a $D$ such that $d_{\mathrm{HCI}} \simeq (C_{\mathrm{HCI}}/C_{\mathrm{DMRG}})^{m_{\mathrm{HCI}}} = D^{m_{\mathrm{DMRG}}/m_{\mathrm{HCI}}} \simeq 6.9 \, D^{3.17}$.

\begin{figure*}
\includegraphics[width=0.99\textwidth]{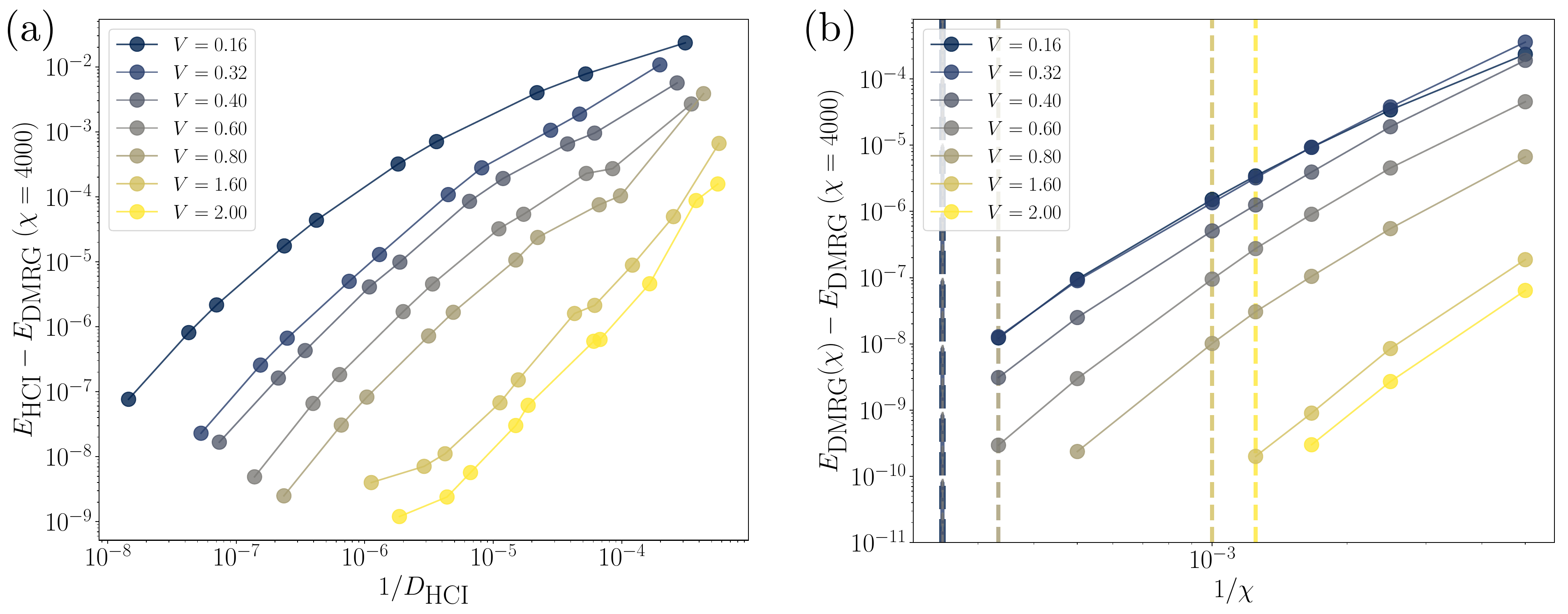}
\caption{Convergence of HCI and DMRG energy estimations as a function of the corresponding cost control parameter for the 4-impurity model. \textbf{(a)}Scaling of the HCI energy error (relative to the best DMRG energy) versus size of the corresponding CI vector for different values of the hybridization amplitude $V$. \textbf{(b)} DMRG energy error versus bond dimension using $E_{\mathrm{DMRG}}(\chi=4000)$ as reference. The vertical dashed lines indicate the value of bond dimension for which the DMRG energy estimation converges.}
\label{fig:accuracy_scaling}
\end{figure*}

\subsection{Compilation of time-evolution circuits}

We use the Jordan-Wigner~\cite{JordanWigner1928} encoding to map the fermionic degrees of freedom into the quantum processor. The second-order Trotter-Suzuki decomposition is used to realize each $|\psi_k\rangle$:
\begin{equation}
    |\psi_k\rangle \approx \left[e^{-i\frac{\Delta t}{2} H_2} e^{-i\Delta t H_1} e^{-i\frac{\Delta t}{2} H_2} \right]^k\ket{\psi_0},
\end{equation}
where $H_1$ and $H_2$ correspond to the one- and two-body terms of the generic interacting-electron Hamiltonian:
\begin{equation}\label{eq: interacting electron ham}
    H = \sum_{\substack{p, q \\ \sigma}} h_{pq} \hat{a}^\dagger_{p\sigma} \hat{a}_{q\sigma}  +  \sum_{\substack{p, q, r, s \\ \sigma \tau}} \frac{h_{pqrs}}{2} \hat{a}^\dagger_{p\sigma} \hat{a}^\dagger_{q\tau} \hat{a}_{s\tau} \hat{a}_{r\sigma}.
\end{equation}
Note that any of the impurity-model Hamiltonians considered in this manuscript can be written in the above form for specific values of $h_{pq}$ and $h_{pqrs}$, where the operators $\hat{a}^\dagger_p$ label both impurity and bath modes.

Since $e^{-i\Delta t H_1} = \exp\left(-i \Delta t \sum_{pq,\sigma} h_{pq} \hat{a}^\dagger_{p\sigma} \hat{a}_{q\sigma} \right)$ is a fermionic Gaussian unitary, it can be realized exactly by a brickwork circuit of Givens rotations, applied to adjacent pairs of qubits, whose depth is equal to the number of fermionic modes in each spin-species, where the qubits are arranged in a one-dimensional chain~\cite{wecker_solving_2015,Kivlichan2018GivensSlater,jiang_correlated_2018} (see Fig.~\ref{fig: impurity bases}). As shown in Fig.~\ref{fig: impurity bases}, $h_{pq}$ is close to diagonal both in the basis where the bath is diagonal, and in the basis of $\mathbf{k}$-adjacent natural orbitals. Additionally, we require $\Delta t \ll 1$, yielding a $e^{-i\Delta t H_1}$ unitary that is close to the identity. This observation is used to approximately compile $e^{-i\Delta t H_1}$ into a depth-3 for the SIAM, and depth-9 for the 4-impurity model brickwork circuit of Givens rotations with adjoint representation $G$, as shown in Fig.~\ref{fig: impurity bases}. The angles of the Givens rotations are obtained by maximizing the Hilbert-Schmidt norm of the $((K+1) \cdot L) \times ((K+1) \cdot L)$ matrix $G^\dagger \cdot \Lambda$, where $\Lambda$ is the matrix exponential of $\left[-i (\Delta t) h\right]$, with $h$ being the $((K+1) \cdot L) \times ((K+1) \cdot L)$ matrix with components $h_{pq}$. The optimization in the space of Gaussian unitaries is performed using gradient descent~\cite{moreno2023orbitalOptims} with the \textit{ADAM}~\cite{kingma2014Adam} update rule. Gradients are computed using the automatic differentiation functionalities of the \texttt{Jax} package~\cite{jax2018github}.

The compilation of $e^{-i\frac{\Delta t}{2} H_2}$ into a quantum circuit is simple when $H_2$ is diagonal, which is the case in the position basis, and the basis with the diagonal bath. In the Jordan-Wigner encoding, this is a controlled-phase gate between the two qubits representing the spin-up and spin-down impurity degrees of freedom. This gate is applied via an auxiliary qubits (green in Fig.~\ref{fig: results} (a) and (d)) that connect the spin-up and spin-down chains of qubits. 

In the basis of $\mathbf{k}$-adjacent natural orbitals, $H_2$ is no longer diagonal. Consequently, its time evolution would naively require circuit depths growing as $\mathcal
{O} \left( [(K+1)\cdot L]^4\right)$. By construction, $H_2$ can be diagonalized by a single-particle basis transformation that mixes only 4 fermionic modes for the SIAM and 12 for the 4-impurity model. This basis transformation can be achieved by the application of Givens rotation gates acting on 4 (or 12) adjacent qubits that implement the Gaussian unitary of the change of basis (see Fig.~\ref{fig: impurity bases}). After this transformation has been applied, $H_2$ is diagonal and its time evolution is realized as described in the previous paragraph. This is depicted in Fig.~\ref{fig: impurity bases}.

For each spin-species, the initial state $\ket{\psi_0}$ is given by the superposition of all possible excitations of the three electrons closest to the Fermi level into the 4 closest empty modes starting from the state $|0000\hdots 01 \hdots 1111\rangle$, and realized by the application of 7 Givens rotation gates of rotation angle $\pi/4$.

%For our experiments, we choose $L = 29$ (see Appendix~\ref{Appending 29}) and $L = 41$ bath sites and consider values of the onsite repulsion: $U = 1, 3, 7, 10$. $L = 29$ requires $60$ qubits under the Jordan-Wigner representation, $30$ for each spin species, while $L = 41$ requires 84 qubits, 42 for each spin species. The projection and diagonalization step are accompanied by the self-consistent configuration recovery procedure introduced in Ref.~\cite{ibm2024chemistry}. In all cases we consider $d = 25$ Trotter steps with $\Delta t = 0.1$.

\subsection{Experiment details}

The experiments were run on IBM Quantum's \textit{ibm\_fez}, a Heron r2 processor with 156 fixed-frequency transmon qubits with tunable couplers on a heavy-hex lattice layout. $L=1$ and $K = 29$ required 61 qubits to implement, while $L=1$ and $K = 41$ required 85 qubits. Each Krylov dimension was sampled with $\num{1e5}$ shots. We implemented $K=29$ on both the basis with diagonal bath and the $\mathbf{k}$-adjacent basis. For the $\mathbf{k}$-adjacent basis, the average median error rates were: readout-error $\num{1.33e-2}$, single-qubit error $\num{2.40e-4}$ and two-qubit gate-error $\num{2.73e-3}$, with the latter two characterized by randomized benchmarking. The average relaxation and dephasing times were $138.5 \ \mu \text{s}$ and $101.25 \ \mu \text{s}$. Likewise, for the basis of diagonal-bath circuits had readout-error $\num{1.48e-2}$, single-qubit error $\num{2.60e-4}$ and two-qubit gate-error $\num{2.77e-3}$. $T_1=131 \ \mu \text{s}$ and $T_2 = 89.25 \ \mu \text{s}$ were slightly lower for these experiments than for the $\mathbf{k}$-adjacent basis ones. The number of two-qubit gates grew linearly with Trotter steps/Krylov dimension with a slope of 312 and 208 for the momentum and $\mathbf{k}$-adjacent basis respectively. The maximum number of two-qubit gates here were 5976 gates and 3281 gates respectively. For $K=41$, we only considered the $\mathbf{k}$-adjacent basis. For these runs, we had average $T_1=131.5 \ \mu \text{s}$ and $T_2 = 96.25 \ \mu \text{s}$. The error rates were: readout-error $\num{1.53e-2}$, single-qubit error $\num{2.6e-4}$ and two-qubit gate-error $\num{2.80e-3}$. The two-qubit gates increased linearly at a rate of 337 gates per Trotter step and the maximum number of two-qubit gates was 6153.  

For the 4-impurity Anderson model circuits, 28 bath sites were considered. Consequently, the number of qubits was 64. The number of two-qubit gates per Trotter step grew at the rate of 613 gates per step, signifying an increase in the density of entangling operations in comparison to the SIAM circuits.  The average relaxation and dephasing times were $T_1 = 155.0 \ \mu\text{s}$ and $T_1 = 101.3 \ \mu\text{s}$ respectively. The error rates were: readout-error $8 \times 10^{-3}$, single-qubit gate error $2.2 \times 10^{-4}$ and two-qubit gate error $3.8 \times 10^{-3}$. The maximum number of gates used was 5222. 

\subsection{Configuration recovery with carry-over}
In the self-consistent configuration recovery scheme introduced in Ref.~\cite{ibm2024chemistry}, the only information that is passed in between self consistent iterations in the average fermionic mode occupancies in the approximated ground state found by SQD and SKQD. We realize that after finding a preliminary approximation to the ground state, additional information may be used to inform the subsequent diagonalization steps. 

In particular, if a basis state $\left|b_j\right\rangle$ is identified as relevant in the description of the ground state at the previous recovery iteration ($\left|\mathbf{\Psi}_j\right|>\tau$), then the state $\left|b_j\right\rangle$ may be included as part of the definition subspaces in subsequent recovery iterations. $\tau$ is a threshold used to determine which basis states have an impact in the the description of the ground state. This strategy is only used in the 4-impurity experiments in this manuscript. We set the value of the threshold to $\tau = 10^{-8}$.

%%%%%%%%%%%%%%%% END OF MAIN TEXT %%%%%%%%%%%%%%%

\newpage
\clearpage

\onecolumngrid

%%%%%%%%%%%%%%%% START OF SUPPLEMENT %%%%%%%%%%%%%%%

% Figures, tables, equations and pages in the supplement are numbered S1, S2 etc.
\renewcommand{\thefigure}{S\arabic{figure}}
\renewcommand{\thetable}{S\arabic{table}}
\renewcommand{\theequation}{S\arabic{equation}}
\renewcommand{\thepage}{S\arabic{page}}
\setcounter{figure}{0}
\setcounter{table}{0}
\setcounter{section}{0}
\setcounter{equation}{0}
\setcounter{page}{1} % not 0 as \newpage already started a supplementary page
% References continue the numbering from the main text.

\begin{center}
\section*{Supplementary Information for: Quantum-Centric Algorithm for Sample-Based Krylov Diagonalization}
\end{center}

\section{Krylov Quantum Diagonalization}\label{sec:kqd}

Let $\mathcal{H}$ denote an $N$-dimensional Hilbert space, where $N=2^n$ and $n$ denotes the number of qubits. Let $H$ be a Hermitian operator on $\mathcal{H}$ with eigenvalues $E_0 \le E_1 \le \dots \le E_{N-1}$ and corresponding orthonormal eigenstates $\ket{\phi_0}, \dots, \ket{\phi_{N-1}}$. Our goal is to estimate $E_0$. In general, the eigenvalue problem for an $N$-dimensional operator is computationally challenging, so we approximate $E_0$ by considering the corresponding eigenvalue problem in a subspace. In particular, given some initial state $\ket{\psi_0}$, consider the \emph{Krylov subspace} $\mathcal{S}$ spanned by
\begin{equation}\label{eq:krylov-basis}
\ket{\psi_k} \coloneq e^{-ikH\,\Delta{t}} \ket{\psi_0}, \quad k \in \{0, 1, \dots, d-1\},
\end{equation}
for some chosen time step $\Delta{t}$ and dimension $d$.

We then solve the following generalized eigenvalue problem
\begin{equation}
\label{eq:gen-eig-app}
\mathbf{\widetilde{H}} v = \widetilde{E} \mathbf{\widetilde{S}} v,
\end{equation}
where $v$ is a coordinate vector corresponding to the state ${\ket{\phi}\coloneq\sum_{k=0}^{d-1} v_k\ket{\psi_k} \in \mathcal{S}}$, and 
\begin{equation}
\label{eq:gen-eig-mats-app}
\mathbf{\widetilde{H}}_{jk} \coloneq \braket{\psi_j | H | \psi_k}, \qquad \mathbf{\widetilde{S}}_{jk} \coloneq \braket{\psi_j | \psi_k}.
\end{equation}

\medskip 

Let 
\begin{equation}\label{eq:delta_ej}
    \Delta E_j = E_j - E_0
\end{equation}
for each $0 < j < N$ and let 
\begin{align}\label{eq:ref-state}
    \ket{\psi_0} = \sum_{k=0}^{N-1} \gamma_k \ket{\phi_k}
\end{align}
be the eigenstate decomposition of the initial state. For completeness, we recall the accuracy of estimating the ground state energy using the Krylov quantum diagonalization approach from \cite{Epperly_2022}. 

\begin{theorem}[Theorem 3.1 in \cite{Epperly_2022}]
\label{thm:krylov}
Let $\Delta{t} = \frac{\pi}{\Delta E_{N-1}}$ and $d$ denote the Krylov dimension as defined in \cref{eq:krylov-basis}. Assume that $d$ is odd; for even $d$ we can use the bound for $d-1$. Then the approximate ground state energy $\widetilde{E}_0$ obtained from the method described in Eqs.~\eqref{eq:krylov-basis}--\eqref{eq:gen-eig-mats-app} satisfies
\begin{equation} 
\label{eq:krylov-error}
0 \le \widetilde{E}_0 - E_0 \le 8 \Delta E_{N-1} \left(\frac{1 - \abs{\gamma_0}^2}{\abs{\gamma_0}^2}\right) \parens{1 + \frac{\pi \Delta E_1}{\Delta E_{N-1}}}^{-(d-1)},
\end{equation}
where $\gamma_0$ denotes the coefficient of $\ket{\phi_0}$ in \cref{eq:ref-state} and $\Delta E_j$ is defined in \cref{eq:delta_ej}.
\end{theorem}

Note that the accuracy in approximating $E_0$ improves exponentially with the dimension $d$ of the Krylov subspace. Moreover, the accuracy is inversely proportional to the overlap $(|\gamma_0|^2)$ of the initial state with the ground state of $H$. We summarize proof steps for \cref{eq:krylov-error} below.

\begin{itemize}
\item First, note that if $d$ is odd, then the projection of the Hamiltonian into the Krylov subspace as we have defined it is the same as the projection of the Hamiltonian into the shifted Krylov space spanned by
\begin{equation}\label{eq:krylov-basis-shifted}
\ket{\psi_k} = e^{-ikH\,\Delta{t}} \ket{\psi_0}, \quad k \in \braces{-\frac{d-1}{2}, -\frac{d-1}{2}+1, \dots, \frac{d-1}{2}-1, \frac{d-1}{2}}.
\end{equation}
This follows because the Gram matrices $\widetilde{\textbf{S}}$ of the two spaces are identical since time evolutions are unitary, and the Hamiltonian projection matrices $\widetilde{\textbf{H}}$ are identical because the time evolutions commute with the full Hamiltonian $H$.
\item For any $0 < a < \pi$ and positive integer $d$, there exists a trigonometric polynomial $p^*$ of degree $d_\text{poly}=\frac{d-1}{2}$ satisfying $p^*(0) = 1$ and
\begin{equation}
\label{eq:trig-poly-bound}
\abs{p^*(\theta)} \le 2(1+a)^{-d_\text{poly}}\quad \forall \theta \in (-\pi, \pi) \setminus (-a, a).
\end{equation}
This polynomial is explicitly constructed as
\begin{equation}
p^*(\theta) = \frac{T_k \parens{1 + 2\frac{\cos\theta - \cos a}{\cos a + 1}} } {T_k\parens{1 + 2 \frac{1 - \cos a}{\cos a + 1}} },
\end{equation}
where $T_k$ is the $k$th-Chebyshev polynomial.

\item Let the Fourier expansion of $p^*$ be
\begin{equation}
p^*\parens{(E - E_0) \Delta t} = \sum_{k=-d_\text{poly}}^{d_\text{poly}} c_k e^{ik E \Delta{t}}.
\end{equation}
Consider the following unnormalized state:
\begin{equation}\label{eq:ansatz_shifted}
\ket{\widetilde\phi_K} = \sum_{k=-d_\text{poly}}^{d_\text{poly}} c_k \ket{\psi_k} = \sum_{k=-d_\text{poly}}^{d_\text{poly}} \sum_{j=0}^{N-1} c_k \gamma_j e^{ik E_j \Delta{t}} \ket{\phi_j}.
\end{equation}
Moreover, from Section~3.1 in \cite{Epperly_2022} we recall that $\sum_k |c_k|^2 \leq 1$. 

\item
The norm of $\ket{\widetilde\phi_K}$ is given by
\begin{equation}
\label{eq:krylov-norm}
\abs{\braket{\widetilde\phi_K|\widetilde\phi_K}}^2 = \sum_{j=0}^{N-1} \abs{\gamma_j}^2 \abs{\sum_{k=-d_\text{poly}}^{d_\text{poly}} c_k e^{ik E_j \Delta{t}}}^2
\ge \abs{\gamma_0}^2 \abs{\sum_{k=-d_\text{poly}}^{d_\text{poly}} c_k e^{ik E_0 \Delta{t}}}^2
= \abs{\gamma_0}^2 \abs{p^*(0)}^2
= \abs{\gamma_0}^2.
\end{equation}

\item Then the energy error of $\ket{\widetilde\phi_K}$ is given by 
\begin{equation}\label{eq:energy_error_shifted}
\frac{\braket{\widetilde\phi_K| (H - E_0) |\widetilde\phi_K}}{\braket{\widetilde\phi_K|\widetilde\phi_K}} 
\le \frac{\sum_{j=0}^{N-1} (E_j-E_0) \abs{\gamma_i}^2 \abs{\sum_{k=-d_\text{poly}}^{d_\text{poly}} e^{ikE_j \Delta{t}}}^2} {\braket{\widetilde\phi_K|\widetilde\phi_K}}
\le \frac{\sum_{j=0}^{N-1} E_j \abs{\gamma_i}^2 \abs{p^*((E_j-E_0)\Delta{t})}^2} {\abs{\gamma_0}^2}.
\end{equation}
Applying (\ref{eq:trig-poly-bound}) with $a = (E_1 - E_0) \Delta{t} = \frac{\pi \Delta E_1}{\Delta E_{N-1}}$ gives
\begin{equation}
p^*((E_j-E_0)\Delta{t}) \le 2\parens{1 + \frac{\pi\Delta{E_1}}{\Delta{E_{N-1}}}}^{-d}
\end{equation}
for each $1 \le j \le N-1$. The $j=0$ term cancels upon taking the difference with $E_0$ for the energy error.

\item Since $\ket{\widetilde\phi_K}$ is explicitly defined in \cref{eq:ansatz_shifted} as an element of the shifted Krylov space \cref{eq:krylov-basis-shifted}, the lowest energy in the shifted Krylov space is upper bounded by the energy of $\ket{\widetilde\phi_K}$.
Finally, as we noted above, the projection of the Hamiltonian into the shifted Krylov space is identical to its projection into the original, unshifted Krylov space, so the lowest energies of the two projections are the same.
Hence, the energy error from finding the lowest energy of the projected Hamiltonian in our Krylov subspace is upper bounded by the energy error \cref{eq:energy_error_shifted} of $\ket{\widetilde\phi_K}$.

\end{itemize}

Thus, for an initial state with a nontrivial ground state overlap $\abs{\gamma_0}^2 = \Theta(1)$ and for a Hamiltonian $H$ that has a well-behaved spectrum, (i.e., $\Delta E_{N-1}$ not growing too quickly and $\Delta E_1$ not too small), a constant error $\widetilde{E}_0 - E_0 \le \eps$ can be achieved by taking $d = O(\log(1/\eps))$.

\section{Proofs and Relevant Details for Theorem~1}
\label{app:proofs}

In this section we prove the performance guarantees for the Krylov diagonalization via quantum unitary sampling approach. As discussed in the main text, we prepare $d$ different states on a quantum computer and collect $M$ samples from each of them by measuring in the computational basis. In particular, let $\ket{\psi_0}$ denote the initial reference state. Then we get $M$ samples from each Krylov basis state $\ket{\psi_k}$ as defined in \cref{eq:krylov-basis}. In particular, we get a sequence of bitstrings $\{a_{km}\}_{k=0}^{M-1}$ for each $k$. Finally, we solve the eigenvalue problem in the subspace spanned by $\{a_{km}~|~k=0,1,\dots,d-1;m=0,1,\dots,M-1\}$.

\medskip

We analyze our algorithm as follows.
We first recall \cref{def:concentration} for convenience:

\noindent
\textbf{\cref{def:concentration}}~($(\alpha_L, \beta_L)$-sparsity).
For any state $\ket{\psi}$, let
\begin{equation}
\label{eq:ground-state-amp-app}
    \ket{\psi} = \sum_{j=1}^Ng_j\ket{b_j},
\end{equation}
where $\{b_j\}$ is some ordering of length-$N$ bitstrings such that $|g_1|\geq |g_2|\geq \cdots |g_N|$.
We say that $\ket{\psi}$ exhibits \emph{$(\alpha_L, \beta_L)$-sparsity} on $\ket{b_1}$ through $\ket{b_L}$ if
\begin{equation}
\label{eq:alpha_def}
    \sum_{j=1}^L |g_j|^2 \geq \alpha_L
\end{equation}
and
\begin{equation}\label{eq:beta_def}
    |g_1|^2, \dots, |g_L|^2 \geq \beta_L.
\end{equation}

Since $|g_i|^2 \geq \beta_L, \forall i$, it implies that   $\alpha_L \ge L \beta_L$, but $\alpha_L$ may be much larger in general. Here, we separate the two as $\alpha_L$ will govern the rate of convergence for successful runs while $\beta_L$ will govern the probability of success.

Below, we restate \cref{mthm:skqd-thm} and prove it using several lemmas in \cref{sec:proof-main-thm}. 

\noindent
\textbf{Theorem \ref{mthm:skqd-thm}.}
\emph{
Let $H$ be a Hamiltonian whose ground state $\ket{\phi_0}$ exhibits $(\alpha_L^{(0)}, \beta_L^{(0)})$-sparsity.
Let $\ket{\widetilde\phi}$ be the lowest energy state supported on the $L$ important bitstrings in $\ket{\phi_0}$.
The error in estimating the ground state energy of $H$ is bounded by
\[
\braket{\widetilde\phi | H | \widetilde\phi} - \braket{\phi_0 | H | \phi_0} \leq \sqrt{8} \norm{H} 
\left(1 - \sqrt{\alpha_L^{(0)}}\right)^{1/2},
\]
provided all $L$ important bitstrings are sampled. The success probability of sampling all $L$ important bitstrings is at least $1 - \eta$ as long as the number of samples from each Krylov basis state exceeds $\left(d^2 \log(L/\eta)\right)/\left(|\gamma_0|^2 (\beta_L^{(0)} - 2\sqrt{\tilde{\varepsilon}})\right)$, where
\[
\tilde{\varepsilon} = 2-2\sqrt{1-\varepsilon/\Delta E_1}
\]
with $\varepsilon$ defined in \cref{meq:eps-kqd}.
}

As discussed in \cref{sec:kqd}, we choose $d$ according to (\ref{eq:krylov-error}) so that the error for the Krylov quantum diagonalization approach is bounded by $\eps$.
If we were to run the Krylov method on the $\{\ket{\psi_k}\}$, we would obtain a state $\ket{\psi} \in \mathcal{S}$ with 
\begin{equation}
\label{eq:krylov-eps}
    0\le\braket{\psi | H | \psi} - \braket{\phi_0 | H | \phi_0} = \widetilde{E}_0 - E_0 < \eps.
\end{equation}
By the proof of Theorem \ref{thm:krylov}, there exists $\ket{\widetilde\phi_K} \in \mathcal{S}$ defined by \eqref{eq:ansatz_shifted}, which is not necessarily the lowest energy state in $\mathcal{S}$, but which satisfies \eqref{eq:krylov-error} and also satisfies the following inequality :
\begin{equation}\label{eq:coeff-property}
\ket{\widehat\phi_K} = \sum_{k=0}^{d-1} d_k \ket{\psi_k}, \qquad \abs{d_k} \le \frac{1}{|\gamma_0|},
\end{equation}
where $\ket{\widehat\phi_K}$ is the normalized version of $\ket{\widetilde\phi_K}$ in \cref{eq:ansatz_shifted}. Therefore, $d_k = c_k/|\langle\widetilde\phi_K|\widetilde\phi_K\rangle|$, which implies that $|d_k| = |c_k|/|\langle\widetilde\phi_K|\widetilde\phi_K\rangle| \leq 1/|\gamma_0|$, where we invoked the condition that $\sum_k |c_k|^2 \leq 1$ (from Section~3.1 in \cite{Epperly_2022}) and $|\langle\widetilde\phi_K|\widetilde\phi_K\rangle|\geq |\gamma_0|$.

Note that to prove \cref{eq:krylov-eps}, we do not actually need to solve the generalized eigenvalue problem to obtain $\ket{\psi}$, but only to have a state satisfying \eqref{eq:krylov-error}, so we will use $\ket{\widetilde\phi_K}$  for our analysis due to its additional structure as in \cref{eq:coeff-property}.

\medskip

In \cref{lem:state-close}, we prove that a state that has a low energy error with respect to the ground state energy is also close to the ground state in 2-norm distance. We then show in \cref{lem:state-conc-close} that if $\ket{\phi_0}$ exhibits sparsity and if $\ket{\psi}$ is close to $\ket{\phi_0}$, then $\ket{\psi}$ also exhibits sparsity.

\begin{lemma}[A state with low energy is close to the ground state]\label{lem:state-close}

Let $H$ be a Hamiltonian with ground state $\ket{\phi_0}$ and spectral gap $\Delta E_1$. If $\ket{\psi}$ is a state such that $\braket{\psi | H | \psi} - \braket{\phi_0 | H | \phi_0} < \eps$ and $\braket{\psi|\phi_0}$ is real, then
\begin{equation}
\norm{\ket{\psi} - \ket{\phi_0}}^2 < 2\left(1 - \sqrt{1 - \frac{\eps}{\Delta E_1}}\right) = O\parens{\frac{\eps}{\Delta E_1}}.
\end{equation}

\end{lemma}

\begin{proof}
Let $\ket{\psi} = \chi_0 \ket{\phi_0} + \chi^\perp\ket{\phi^\perp}$, where $\ket{\phi^\perp}$ is normalized and orthogonal to $\ket{\phi_0}$, such that ${\abs{\chi_0}^2 + \abs{\chi^\perp}^2 = 1}$. We can write $\braket{\psi | H | \psi} - \braket{\phi_0 | H | \phi_0} < \eps$ as
\begin{align}
\eps > \abs{\chi_0}^2 E_0 + \abs{\chi^\perp}^2 \braket{\phi^\perp | H | \phi^\perp} - E_0 
& = (\abs{\chi_0}^2 - 1) E_0 + (1 - \abs{\chi_0}^2) \braket{\phi^\perp | \mathcal{H} | \phi^\perp} \\
& \ge (1 - \abs{\chi_0}^2) \Delta E_1.
\end{align}
Thus $\abs{\chi_0}^2 > 1 - \frac{\eps}{\Delta E_1}$, giving
\begin{equation}
\norm{\ket{\psi} - \ket{\phi_0}}^2 = \abs{\chi_0 - 1}^2 + \abs{\chi^\perp}^2 
= 2 - 2\chi_0
< 2 - 2\sqrt{1 - \frac{\eps}{\Delta E_1}}.
\end{equation}
Expanding to the first order in $\frac{\eps}{\Delta E_1}$ gives an error bound of $\frac{\eps}{\Delta E_1} + O(\eps^2)$.
\end{proof}

\begin{lemma}\label{lem:state-conc-close}
If $\ket{\phi_0}$ exhibits $(\alpha_L^{(0)}, \beta_L^{(0)})$-sparsity and $\norm{\ket{\psi} - \ket{\phi_0}}^2 < \eps$, then $\ket{\psi}$ exhibits $(\alpha_L, \beta_L)$-sparsity for
\begin{equation}
\label{eq:modified_alpha_beta}
\alpha_L = \alpha_L^{(0)} - 2\sqrt{\eps} \qquad \beta_L = \beta_L^{(0)} - 2\sqrt{\eps}.
\end{equation}
\end{lemma}

\begin{proof}
Expand $\ket{\psi}$ and $\ket{\phi_0}$ in the computational basis as $\ket{\psi} = \sum_{j=1}^N a_j \ket{b_j}$ and $\ket{\phi_0} = \sum_{j=1}^N c_j \ket{b_j}$, respectively. Then
\begin{align}
\sum_{j=1}^L \abs{a_j}^2 = \sum_{j=1}^L \abs{c_j + (a_j - c_j)}^2 &= \sum_{j=1}^L \abs{c_j}^2 + \sum_{j=1}^L \abs{a_j - c_j}^2 + 2\sum_{j=1}^L \Re\parens{c_j (a_j - c_j)} \\
& \ge \alpha_L^{(0)} + 0 - 2 \sum_{j=1}^L \abs{c_j} \abs{a_j - c_j} \\
& \ge \alpha_L^{(0)} - 2 \sqrt{\sum_j \abs{c_j}^2 \sum_k \abs{a_k - c_k}^2} \\
& \ge \alpha_L^{(0)} - 2\sqrt{\eps},
\end{align}
showing that we may take $\alpha_L = \alpha_L^{(0)} - 2\sqrt{\eps}$. In the last line we replaced $\sum_{j=1}^L |c_j|^2 \leq 1$ and used the fact that $\norm{\ket{\psi} - \ket{\phi_0}}^2 < \eps$.

We now prove a similar condition for $\beta_L$. Since $\norm{\ket{\psi} - \ket{\phi_0}}^2 < \eps$, we must have $\abs{a_j - c_j} < \sqrt{\eps}$ for each $1 \le j \le N$. 
Then for each $1 \le j \le L$, we have
\begin{align}
\abs{a_j}^2 = \abs{c_j + (a_j - c_j)}^2 &= \abs{c_j}^2 + \abs{a_j - c_j}^2 + 2 \Re(c_j (a_j - c_j)) \\
&\ge \beta_L^{(0)} + 0 - 2\abs{c_j} \abs{a_j - c_j} \\
&\ge \beta_L^{(0)} - 2\sqrt{\eps}.
\end{align}
Thus we may take $\beta_L = \beta_L^{(0)} - 2\sqrt{\eps}$.
\end{proof}

We have shown in \cref{lem:state-conc-close} that if the Krylov quantum diagonalization (KQD) approach has $\varepsilon$ energy error and if the true ground state exhibits $(\alpha_L^{(0)}, \beta_L^{(0)})$-sparsity, then a solution to the KQD approach exhibits $(\alpha_L, \beta_L)$-sparsity for $\alpha_L, \beta_L$ given by \eqref{eq:modified_alpha_beta}. Using this result, we now prove in \cref{lem:nontrivial-bitstring} and \cref{lem:pfail} that with a high probability, each bitstring (among the $L$ important bitstrings) corresponding to the ground state $\ket{\phi_0}$ appears nontrivially in at least one of the Krylov basis states. This is crucial for establishing a bound on the error corresponding the Krylov diagonalization via quantum unitary sampling approach.

\begin{lemma}
[Each relevant bitstring appears nontrivially in at least one Krylov basis state]
\label{lem:conc-bitstring-prob} 
Let $\ket{\widehat\phi_K}$ satisfy $(\alpha_L, \beta_L)$-sparsity as in \cref{def:concentration} and $\ket{\widehat\phi_K} = \sum\limits_{k=0}^{d-1} d_k \ket{\psi_k}$ as in \cref{eq:coeff-property}. If each $\ket{\psi_k}$ is represented in the computational basis as
\(
\ket{\psi_k} = \sum_{j=1}^{N} c^{(k)}_{j} \ket{b_j}
\) for each $k=0,1,...,d-1$, then for each $j=0,1,\dots,L-1$ there exists some $c^{(k)}_j$ such that $\abs{c^{(k)}_{j}}^2 \ge p$ with
\begin{equation}
p = \frac{|\gamma_0|^2 \beta_L}{d^2}.
\end{equation}
Here, $|\gamma_0|^2$ denotes the overlap of the KQD initial state $\ket{\psi_0}$ with the ground state, as defined in \cref{eq:krylov-basis}. 
\end{lemma}\label{lem:nontrivial-bitstring}

\begin{proof}
Let $\ket{\widehat{\phi}_K} = \sum_{j=1}^N a_j \ket{b_j}$ in the computational basis. For each $1 \le j \le L$, we have
\begin{equation}\label{eq:coeff-large}
\sum_{k=0}^{d-1} \abs{c^{(k)}_{j}} \ge \sum_{k=0}^{d-1} |\gamma_0|\abs{d_k} \abs{c^{(k)}_{j}} 
\ge |\gamma_0|\abs{\sum_{k=0}^{d-1} d_k c^{(k)}_{j}} 
= |\gamma_0| \abs{a_j}
\ge |\gamma_0| \sqrt{\beta_L},
\end{equation}
where the first inequality follows from $|d_k||\gamma_0|\leq 1$, as in \cref{eq:coeff-property}, and
the last inequality follows from \cref{lem:state-conc-close}.
Therefore, \cref{eq:coeff-large} implies that there must be some $k$ for which $\abs{c^{(k)}_{j}} \ge \frac{|\gamma_0| \sqrt{\beta_L}}{d}$, or equivalently $\abs{c^{(k)}_{j}}^2 \ge \frac{|\gamma_0|^2 \beta_L}{d^2}$, as desired.
\end{proof}

\begin{lemma}\label{lem:pfail}
If we make $M$ measurements from each $\ket{\psi_k}$, the probability of not obtaining all $\ket{b_0},\ket{b_1},...,\ket{b_{L-1}}$ among the sampled bitstrings is bounded by
\begin{equation}
p_{\operatorname{fail}} \le L (1-p)^M \le L e^{-Mp}.
\end{equation}
\end{lemma}

\begin{proof}
The $M$ measurements are independent, so the probability of not obtaining a particular bitstring $\ket{b_j}$ for $j\in\{0,1,...,L-1\}$ is $(1-p)^M$. The probability of not obtaining at least one of the $L$ bitstrings follows from union bound.
\end{proof}

We now recall a result from \cite{ibm2024chemistry} relating the energy of a state defined in a subspace to the ground state energy on the full $n$-qubit space. 

\begin{lemma}[Appendix B.1 from \cite{ibm2024chemistry}]
\label{lem:conc-error-bound}
Let $\ket{\widetilde\phi} = \frac{1}{C} \sum_{j=0}^{L-1} c_j \ket{b_j}$, where $C = \sqrt{\sum_{j=0}^{L-1} \abs{c_j}^2}$ is the normalization constant and $\sum_{j=0}^{N-1} c_j \ket{b_j}$ defines the ground state in the computational basis in decreasing order of coefficient magnitude, as in \cref{def:concentration}. This state has energy close to the ground state energy, with difference bounded by
\begin{equation}
\braket{\widetilde\phi | H | \widetilde\phi} - \braket{\phi_0 | H | \phi_0} \le 2\sqrt{2} \norm{H} \parens{ 1 - \sqrt{\alpha_L^{(0)}} }^{1/2}.
\end{equation}
\end{lemma}

\begin{proof}
First, we rewrite the error in energy. Let $\ket{\phi'} = \ket{\widetilde\phi}-\ket{\phi_0} $. Then we get 
\begin{equation}
\braket{\widetilde\phi | H | \widetilde\phi} - \braket{\phi_0 | H | \phi_0}
= \langle{\widetilde\phi}\vert H |\phi'\rangle + \langle \phi' \vert H  |\phi_0\rangle~.
\end{equation}
Then we have
\begin{equation}
\label{eq:partial_bound}
\braket{\widetilde\phi | H | \widetilde\phi} - \braket{\phi_0 | H | \phi_0}
\le |\bra{\widetilde\phi} H |\phi'\rangle|+ \abs{\langle \phi' \vert H  |\phi_0\rangle}
\le 2 \norm{\ket{\widetilde\phi}} \norm{H |\phi'\rangle}
\le 2 \cdot 1 \cdot \norm{H} \norm{|\phi'\rangle}.
\end{equation}

Finally, we calculate the norm difference as
\begin{equation}
\norm{\ket{\widetilde\phi} - \ket{\phi_0}}^2 
= \sum_{j=0}^{L-1} \parens{\frac{1}{C} - 1}^2 \abs{c_j}^2 + \sum_{L}^{N-1} \abs{c_j}^2 
= \parens{1 - \frac{2}{C} + \frac{1}{C^2}} C^2 + (1-C^2)
= 2 - 2C,
\end{equation}
and we have $C \ge \sqrt{\alpha_L^{(0)}}$ by \eqref{eq:alpha_def}, so $\norm{\ket{\widetilde\phi} - \ket{\phi_0}} \le \sqrt{2} \parens{1 - \sqrt{\alpha_L^{(0)}}}^{1/2}$.
Since $\ket{\phi'}$ is defined to be $\ket{\widetilde\phi} - \ket{\phi_0}$, plugging the above into \eqref{eq:partial_bound} completes the proof.
\end{proof}

\subsubsection{Proof of \cref{mthm:skqd-thm}}\label{sec:proof-main-thm}

From Theorem \ref{thm:krylov}, we have a state $\ket{\psi}$ with $\braket{\psi | H | \psi} - \braket{\phi_0 | H | \phi_0} < \eps$ with
\begin{equation}
\eps = 8 \Delta E_{N-1} \left(\frac{1 - \abs{\gamma_0}^2}{\abs{\gamma_0}^2}\right) \parens{1 + \frac{\pi \Delta E_1}{\Delta E_{N-1}}}^{-(d-1)}.
\end{equation}
Let 
\begin{align}\label{eq:eps-prime}
    \varepsilon' = \sqrt{1-\varepsilon/\Delta E_1}.
\end{align} 
Then, from \cref{lem:state-close}, we get 
\begin{equation}
\norm{\ket{\psi} - \ket{\phi_0}}^2 < 2 - 2\varepsilon'.
\end{equation}

Given $\ket{\phi_0}$ exhibits $(\alpha_L^{(0)}, \beta_L^{(0)})$-sparsity, by Lemma \ref{lem:state-conc-close} $\ket{\psi}$ exhibits $(\alpha_L, \beta_L)$-sparsity with parameters
\begin{equation}\label{eq:betal_supp}
\alpha_L = \alpha_L^{(0)} - 2\sqrt{2 - 2\varepsilon'} \qquad
\beta_L = \beta_L^{(0)} - 2\sqrt{2 - 2\varepsilon'}
\end{equation}
Hence by Lemma \ref{lem:conc-bitstring-prob}, for each of the $L$ important bitstrings $b_j$ with $j=0,1,...,L-1$, we will be able to sample $b_j$ with probability at least
\begin{equation}\label{eq:final-p}
p = \frac{\abs{\gamma_0}^2}{d^2} \parens{\beta_L^{(0)} - 2\sqrt{2 - 2\varepsilon'}}
\end{equation}
from at least one of the $\ket{\psi_k}$. Given $M$ measurements, the probability of failing to sample $b_j$ is at most $(1-p)^M$. Repeating this for each bitstring in 
$\{b_j\}_{j=0}^{L-1}$, the probability of failing to sample all $L$ bitstrings $b_1$ through $b_L$ is
\begin{equation}
p_{\text{fail}} \le L(1-p)^M \le Le^{-Mp} = L \exp\parens{-\frac{M\abs{\gamma_0^2}}{d^2} \parens{\beta_L^{(0)} - 2\sqrt{2 - 2\sqrt{1 - \frac{\eps}{\Delta E_1}}}}}
\end{equation}
by union bound and by using $p$ from \cref{eq:final-p} and $\varepsilon'$ from \cref{eq:eps-prime}. If we succeed in sampling all $L$ bitstrings, then the state $\ket{\widetilde\phi}$ as defined in Lemma \ref{lem:conc-error-bound} exists in the sampled subspace, so the calculated energy will be bounded by
\begin{equation}
\braket{\widetilde\phi | H | \widetilde\phi} - \braket{\phi_0 | H | \phi_0} \le 2\sqrt{2} \norm{H} \parens{ 1 - \sqrt{\alpha_L^{(0)}} }^{1/2},
\end{equation}
which completes the proof of \cref{mthm:skqd-thm}.

\section{Sparsity in the Ising model}\label{sec:sparsity-ising}

In this section, we prove the sparsity of the ground state in the computational basis for a particular Hamiltonian. We consider the transverse field Ising model with periodic boundary conditions
\begin{equation}
    H_n(h) = -\sum_{i=0}^{n-1} Z_i Z_{i+1} - h \sum_{i=0}^{n-1} X_i.
\end{equation}

\begin{theorem}
    \label{thm:tfim-sparsity}
    If $h = O((k/n)^a)$ for any $a > 1/2$, then in the limit $n\rightarrow\infty$ the ground state of $H_n(h)$ is fully supported on the $O(n^k)$ $Z$-basis states with Hamming weight at most $k$.
\end{theorem}
\begin{proof}
We will ignore the degeneracy of the ground state for simplicity, but the result holds without this assumption by symmetry.
Let $\ket{\phi_n(h)}$ be the ground state of $H_n(h)$.
With this assumption, the ground state is $\ket{\phi_n(0)} = \ket{00\dots0}$ and $\lim_{h\to\infty}\ket{\phi_n(h)} = \ket{++\dots +}$.
Thus, intuitively, $h$ controls the sparsity of the ground state in the computational basis.

Let $|x|$ be the Hamming weight of the bit string $x$, and define
\begin{align}
    &M_n(h) = \frac{1}{n} \sum_{i=0}^{n-1} \bra{\phi_n(h)} Z_i \ket{\phi_n(h)} \\
    &M(h) = \lim_{n\to\infty} M_n(h)\\
    &\bar S_n(k, h) = \sum_{\substack{x \in \{0,1\}^n \\ |x| \leq k}} \abs{\braket{x|\phi_n(h)}}^2 \\
    &S_n(k, h) = 1 - \bar S_n(k, h).
\end{align}
$S_n(k, h)$ is a proxy for sparsity -- it being small implies that there is very little weight on states outside of the $d$-dimensional subspace defined by Hamming weight less than $k$, where $d = \sum_{w=0}^k \binom{n}{w}$.

Define $\alpha_x$ by $\ket{\phi_n(h)} = \sum_{x\in \{0,1\}^n} \alpha_x \ket x$.
Furthermore, define $\bar P_n(w, h) = \sum_{\substack{x\in \{0,1\}^n \\ |x| = w}} |\alpha_x|^2$, so that $\bar S_n(k, h) = \sum_{w=0}^k \bar P_n(w, h)$.
Then 
\begin{align}
    M_n(h)
    &= \frac{1}{n}\sum_{i=0}^{n-1} \sum_{x,y\in \{0,1\}^n} \bar\alpha_x \alpha_y \bra x Z_i \ket y \\
    &= \frac{1}{n}\sum_{i=0}^{n-1} \sum_{x\in \{0,1\}^n} |\alpha_x|^2 (-1)^{x_i} \\
    &= \frac{1}{n} \sum_{w=0}^n \sum_{\substack{x\in \{0,1\}^n \\ |x|=w}} |\alpha_x|^2 \sum_{i=0}^{n-1} (-1)^{x_i} \\
    &= \frac{1}{n} \sum_{w=0}^n \sum_{\substack{x\in \{0,1\}^n \\ |x|=w}} |\alpha_x|^2  (n-2w) \\
    &= \sum_{w=0}^n \bar P_n(w, h) (1-2w/n) \\
    &= \sum_{w=0}^k \bar P_n(w, h) (1-2w/n) + \sum_{w=k+1}^{n} \bar P_n(w, h) (1-2w/n)\\
    &\leq \bar S_n(k, h) + \sum_{w=k+1}^{n} \bar P_n(w, h) (1-2w/n) \\
    &\leq \bar S_n(k, h) + (1-2(k+1)/n) \sum_{w=k+1}^{n} \bar P_n(w, h)  \\
    &= \bar S_n(k, h) + (1-2(k+1)/n) S_n(k, h)  \\
    &= 1 - \frac{2(k+1)}{n} S_n(k, h).
\end{align}
It follows that 
\begin{equation}
    \label{eq:sparsity}
    S_n(k, h) \leq \min\left(\frac{n(1-M_n(h))}{2k+2}, 1 \right) .
\end{equation}

If $S_n(k, h) \to 0$ as $n\to\infty$, then the subspace of Hamming weight $\leq k$ bitstrings is fully capturing the ground state.
The dimension of this space is $d = \sum_{w=0}^k \binom{n}{w}$, which is $\sim n^k$ for constant $k$.

\cref{eq:sparsity} implies that if $M_n(h) = 1 - O((k/n)^a)$ for any $a > 1$, the ground state is fully supported on $O(n^k)$ bitstrings (ie.~if $M_n(h) = 1 - O((k/n)^a)$ for $a>1$, then $S_n(k, h)$ decays to zero with increasing $n$).
From the phase diagram of the transverse field Ising model, we know \cite[Eq.~3.12]{pfeuty1970the-one-dimensi} that $M(h) = (1-h^2)^{1/8}$ for $0 \leq h \leq 1$, and $M_n(h) \to M(h)$ converges continuously with $n\to\infty$.
It follows that $h^2 = O((k/n)^a)$ for $a > 1$ suffices.
This completes the proof.
\end{proof}

\begin{corollary}
    \label{cor:tfim-sparsity-disordered}
    If $h = O((n/k)^a)$ for any $a > 1/2$, then in the limit $n\rightarrow\infty$ the ground state of $H_n(h)$ is fully supported on $O(n^k)$ $X$-basis states.
\end{corollary}
\begin{proof}
    Apply a Hadamard matrix on each qubit to get a transformed Hamiltonian 
    $H_n'(h) = - \sum_{i=0}^{n-1} X_i X_{i+1} - h \sum_{i=0}^{n-1} Z_i$.
    Then, by \cref{thm:tfim-sparsity}, as $n\rightarrow\infty$ the ground state of $H_n'(h)$ is fully supported on $O(n^k)$ $X$-basis states if $1/h = O((k/n)^a)$ for any $a > 1/2$.
\end{proof}

\cref{thm:tfim-sparsity} proves that the ground state of the transverse field Ising model is sparse in a product state basis (the $Z$-basis) deep in the ordered phase.
Meanwhile, \cref{cor:tfim-sparsity-disordered} proves that the ground state is sparse in a product state basis (the $X$-basis) deep in the disordered phase.

\section{Comparison with Alternative Notions of Sparsity}\label{sec:app-sparsity-notions}

 Note that our definition of sparseness (or peakedness), as defined in \cref{eq:alpha_def} and \cref{eq:beta_def} differs from those in \cite{bravyi2024classical, aaronson2024verifiable}. In \cite{aaronson2024verifiable}, a unitary circuit \( C \) is defined as \(\delta\)-peaked if there exists at least one bitstring \( s \in \{0,1\}^n \) such that \( |\langle s |C|0\rangle|^2 \geq \delta \). The focus of \cite{aaronson2024verifiable} was on random peaked circuits and determining whether such circuits can be distinguished from fully random circuits in classical polynomial time. Similarly, \cite{bravyi2024classical} defines a circuit family \( \{U_n\} \) as peaked if for some \( a \in \mathbb{Z}_{\geq 0} \), each \( U_n \) is \(\delta\)-peaked for \( n \geq 0 \) with \( \delta = n^{-a} \). In this context, an \( n \)-qubit circuit is considered peaked if it has an output probability that is at least inverse-polynomial in \( n \). In \cite{bravyi2024classical}, the authors developed classical algorithms for sampling and estimating output probabilities of constant-depth peaked quantum circuits. In contrast, our definition of sparsity requires the weight to be concentrated on $L$ bitstrings, rather than requiring at least one bitstring to have high probability.

\section{Comparison between the sample complexity in SKD and KQD in the Transverse Field Ising Model}\label{sec:kqd-skqd}
To establish an insight on the practical performance of SKQD, we first showcase numerical simulations on a lattice model. These simulations use the shifted Krylov space given by \cref{eq:krylov-basis-shifted} and the usual $\Delta{t} = \pi/\Delta{E_{N-1}}$ in the context of \cref{thm:krylov}.

We consider a perturbed transverse field Ising model
\begin{equation}
\label{eq:ising}
H = -\sum_{j=1}^{n-1} Z_j Z_{j+1} - h_1 \sum_{j=1}^n X_j - h_2 Z_1.
\end{equation}
When $h_1 = h_2 = 0$, the ground states are spanned by the bitstrings $\ket{0^n}$ and $\ket{1^n}$. A positive $h_2$ breaks the degeneracy in favor for $\ket{0^n}$. In \cref{sec:sparsity-ising}, for $h_1 = h$ and $h_2=0$, we show that if $h = O((k/n)^a)$ for any $a > 1/2$, then the ground state of $H_n(h)$ is fully supported on $O(n^k)$ bitstrings.

\begin{figure}[t]
\leavevmode\centering
\includegraphics[width=0.5\columnwidth]{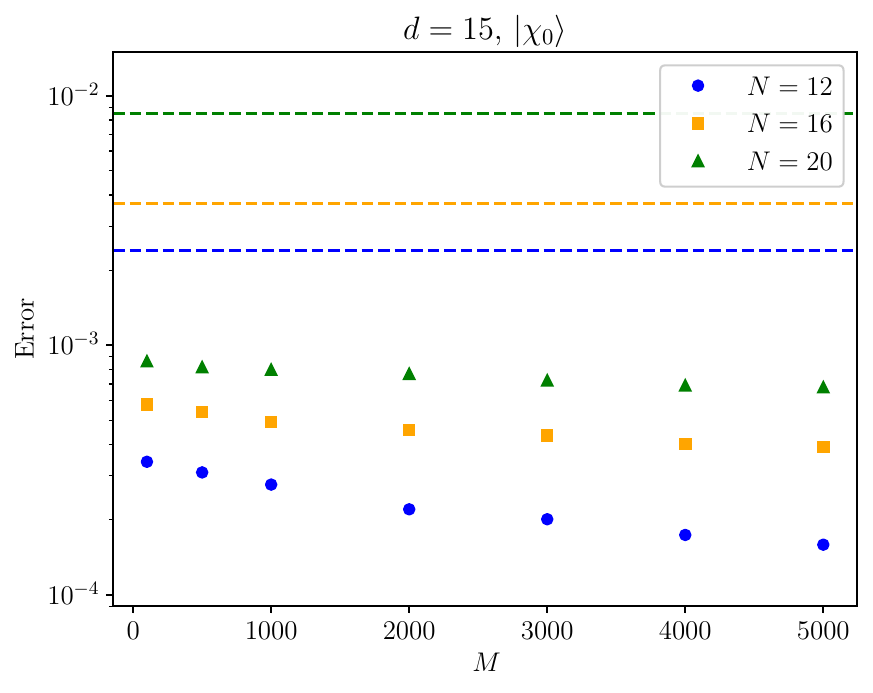}
\caption{\textbf{Comparison of SKQD and KQD methods.} For a perturbed transverse field Ising model Hamiltonian $H$ with equal transverse field and perturbation $h_1 = h_2 = 0.1$, and initial state $\ket{\chi_0} = \ket{0^n}$, the SKQD approach (markers) achieves lower error compared to the standard KQD method. We evaluate the SKQD approach with varying numbers of samples measured per Krylov state, denoted by $M$, and set the number of Krylov basis states to $d=15$. In the KQD approach (dashed lines), we incorporate Gaussian noise with a standard deviation of $\frac{1}{\sqrt{M}}$, where $M = 5000$, while estimating matrix elements of $H$.}
\label{fig:krylov-comparison}
\end{figure}

For our numerical simulations, we use the initial state $\ket{\chi_0} = \ket{0^n}$.  In Figure \ref{fig:krylov-comparison}, we compare the performance of SKQD approach with the standard KQD approach. To ensure a fair comparison, we add a Gaussian noise $\mathcal{N}(0, \frac{1}{\sqrt{M}})$ to each matrix element, as described in \cref{eq:gen-eig-mats}. Let $h_1 = h_2 = 0.1$. We run the SKQD approach for $d=15$ different Kyrlov basis states and for varying numbers of samples from each basis state. We perform simulations with different numbers of qubits, as shown in \cref{fig:krylov-comparison}. We  set $M=5000$ in the KQD approach while computing each matrix element. Moreover, we selected the best instance for the SKQD approach from 1000 trials.
\cref{fig:krylov-comparison} demonstrates that our approach (SKQD) outperforms the standard KQD approach across different numbers of qubits.  Thus, our numerical simulations extend beyond the analytical results, showing that the SKQD approach can outperform the standard KQD approach under the sparsity assumption on the ground states.

\section{Effect of Trotter error in SKQD}\label{sec:trotter_error}

In this section, we briefly analyze the effect of Trotter error on the performance of SKQD. We assume that $U_k \equiv e^{-i k \Delta t H}$ can be approximated with $\gamma$ Trotter error for all $k \in \{0, 1, \dots, d-1\}$. Let $V_k$ denote the Trotter approximation for $U_k$. Then the following inequalities hold: 
\begin{align}
    \Vert U_k - V_k \Vert &\leq \gamma,\\
\implies \Vert  U_k \ket{\psi_0} - V_k \ket{\psi_0}\Vert_1 &\leq \gamma,\\
\implies \Vert p_k - \hat{p}_k \Vert_1 &\leq \gamma, \\
\implies |p_k^{(j)}- \hat{p}_k^{(j)}| & \leq \gamma,\\
\implies \hat{p}_k^{(j)}&\geq p_{k}^{(j)} - \gamma, \label{eq:noisy-bstring-prob}
\end{align}
where the third inequality follows from monotonicity of the trace distance. Here, $p_k (\hat{p}_k)$ is the distribution on computational basis state for the state $U_k \ket{\psi_0} (V_k \ket{\psi_0})$. Moreover $q_k^{(j)}$ denotes the probability of bitstring $j$ base on distribution $q_k$. 

Then by combining \cref{eq:noisy-bstring-prob} with \cref{lem:conc-bitstring-prob} and \cref{lem:pfail}, we get that the number of shots ($M$) needed for sampling all relevant $L$ bitstrings with a high probability (larger than $1-\eta$) scales as
\begin{align}
    M \geq \frac{\log L/\eta}{\left( |\gamma_0|^2\beta_L/d^2 -\gamma\right) }~.
\end{align}

\section{Additional SKQD experiments for the Single Impurity Anderson Model}\label{msec:ksqsd-experiments}

The aim of this section is to study the accuracy of SKQD for in the SIAM for different system sizes than the one shown in the main text. In particular we consider the $K = 29$ bath-site model. As in the main text, we benchmark the accuracy of SKQD against DMRG. Each DMRG run performed 20 sweeps. The first four sweeps have a maximum bond dimension of $250$, the next four sweeps have a maximum bond dimension of $400$, and the remaining $12$ sweeps a maximum bond dimension of $500$. At each sweep we add noise of amplitude $10^{-4}$ in the first four sweeps, $10^{-5}$ in the next four sweeps, $10^{-7}$ in the next four sweeps, and $0$ in the remaining. We analyze the relative error in the ground state energy and the agreement in the prediction of the two-point spin and density correlation functions (see Eqs.~\ref{eq: spin correlation} and~\ref{eq: density correlation}).

\begin{figure*}[t]
\centering
\includegraphics[width=1\linewidth]{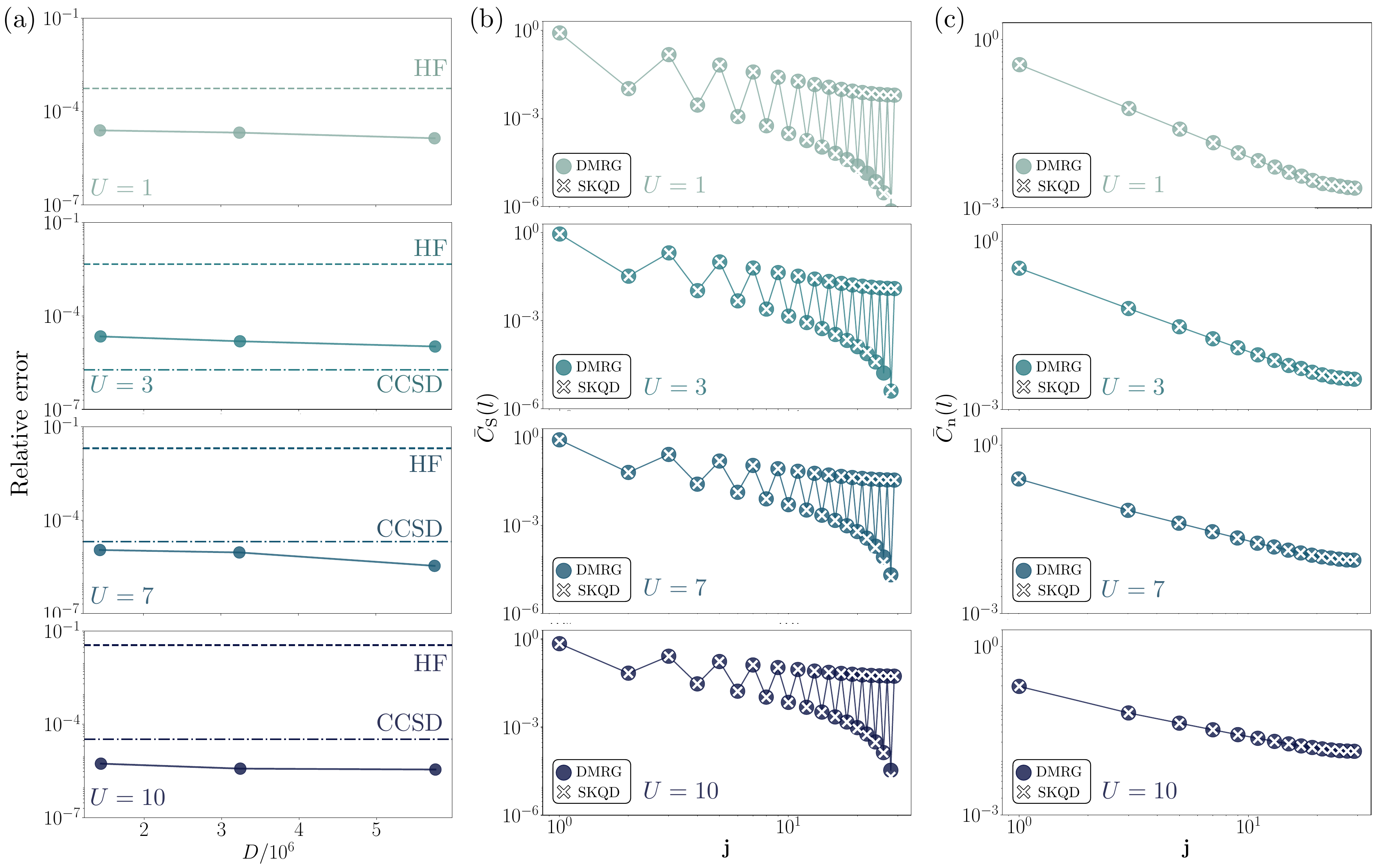}
\caption{SKQD vs DMRG in the SIAM with 29 bath sites (61-qubit experiment). Different rows correspond to different values of the onsite repulsion $U$ in the impurity. \textbf{(a)} Relative error in the ground state energy estimation using SKQD, as a function of the subspace dimension $D$. The DMRG estimation is taken as the ground truth. The Hartree-Fock (HF) and coupled cluster with single and double excitations (CCSD) are also included for reference. The dots correspond to the SKQD estimation in the $\mathbf{k}$-adjacent natural orbitals. \textbf{(b)} Comparison of the two-point spin correlation functions (see Eq.~\ref{eq: spin correlation}) obtained with DMRG and SKQD. \textbf{(c)} Comparison of the two-point density correlation functions (see Eq.~\ref{eq: density correlation}) obtained with DMRG and SKQD. }
\label{fig: dmrg vs sqd appendix}
\end{figure*}

Figure~\ref{fig: dmrg vs sqd appendix} (a) shows the relative (to DMRG) error in the SKQD ground-state energy estimation as a function of the subspace dimension on the SKQD eigenstate solver $D$. The Hartree-Fock (HF) and coupled cluster with single and double excitations (CCSD) errors are also shown for reference. The SKQD relative error decreases from values $\sim 10^{-5}$ to $\sim 10^{-6}$ as $U$ increases from $U = 1$ to $U = 10$. This is a consequence of the increased ground-state sparsity for larger values of $U$. 

Panel (b) of Fig.~\ref{fig: dmrg vs sqd appendix} compares the values of $\bar{C}_\textrm{S}(\mathbf{j})$ obtained from SKQD to those obtained with DMRG. The SKQD estimations are in excellent agreement with the DMRG values for most values of $\mathbf{j}$, the distance between the impurity spin and the bath spin. There are small deviations for odd values of $\mathbf{j}$ at larger values of $\mathbf{j}$, where the value of the correlation is negligible. Panel (c) in Fig.~\ref{fig: dmrg vs sqd appendix} compares the values of $\bar{C}_\textrm{n}(\mathbf{j})$ obtained from SKQD to those obtained with DMRG, for even values of $\mathbf{j}$. The SKQD estimations are in excellent agreement with the DMRG values for all values of $\mathbf{j}$.

The accuracy for the system size presented in this Appendix ($K = 29$), does not significantly differ from the accuracy on the larger system size ($K = 41$) shown in the main text. We conclude that the accuracy of SKQD does not significantly deteriorate with system size in the SIAM.

\section{Signal in the quantum experiments for the Single Impurity Anderson Model}\label{signal}
Given the large circuit sizes of our experiments and the effect of noise, we investigate whether there is a useful signal coming out of the quantum circuits, comparing the outcome of SKQD (with configuration recovery) run on samples coming from the device and uniform random samples. The random samples are drown from the uniform distribution in the space of bitstrings whose length is the same as the number of fermionic modes in the system.

This test is conducted on the SIAM with $K = 41$ bath sites and the same values of $U$ as the ones shown in the main text: $U = 1, 3, 7, 10$. The subspace dimension chosen to project and diagonalize the Hamiltonian is $D = 2.56\cdot 10^6$ electronic configurations, and the total number of sampled bitstrings is the same in both cases: $2.5\cdot 10^6$.
\begin{figure*}[t]
\centering
\includegraphics[width=.9\linewidth]{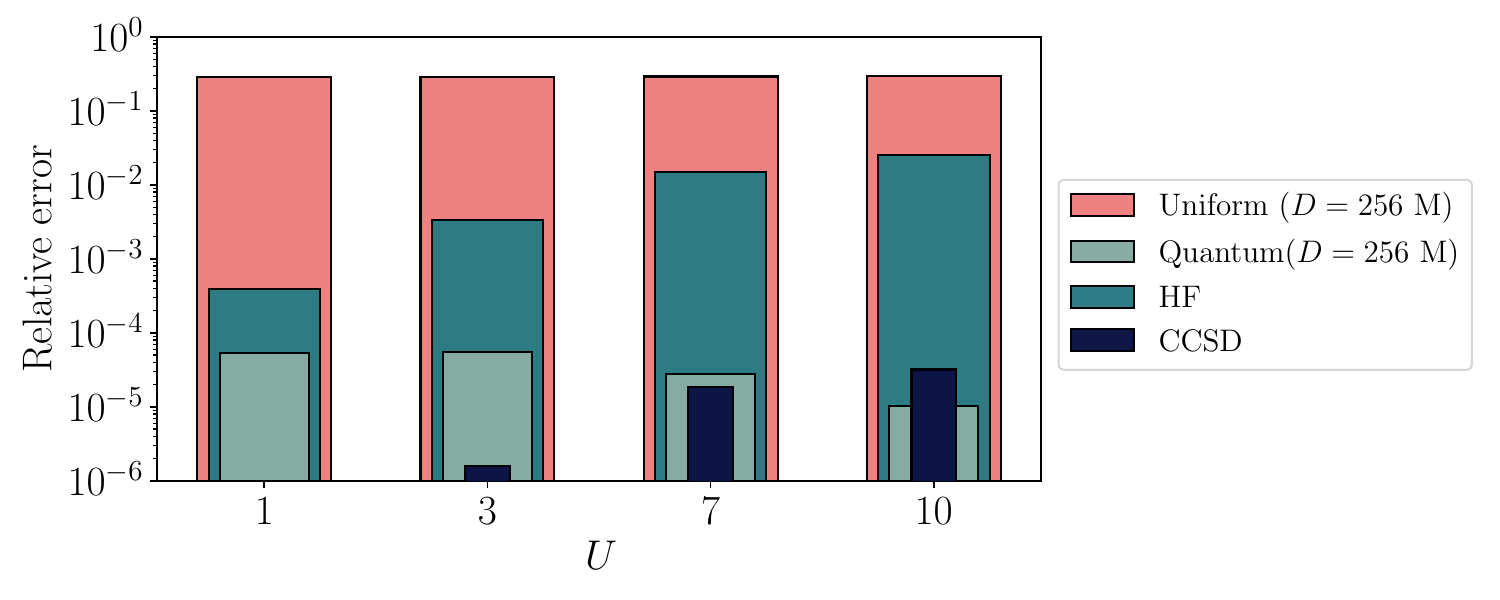}
\caption{Signal in the quantum experiments. Relative error in the SKQD ground-state energy estimation for different values of $U$. The error is computed relative to the DMRG energy. SKQD is run both on samples coming from the quantum device and the uniform distribution. The HF and CCSD relative errors are included for reference.}
\label{fig: skqd vs uniform}
\end{figure*}

Figure~\ref{fig: skqd vs uniform} shows the relative (to DMRG) error in the ground-state energy obtained from running SKQD on samples drawn from the device and samples drawn from the uniform distribution for different values of $U$. The relative error in the ground-state energy is orders of magnitude lower in the SKQD estimation run on samples coming from the quantum device. We therefore conclude that SKQD with configuration recovery is capable of extracting an useful signal from the device.

% \newpage
% \clearpage

\bibliographystyle{unsrt}
\bibliography{refs}

\begin{thebibliography}{10}

\bibitem{kitaev1995quantum}
A~Yu Kitaev.
\newblock Quantum measurements and the {A}belian stabilizer problem.
\newblock {\em arXiv preprint quant-ph/9511026}, 1995.

\bibitem{peruzzo2014variational}
Alberto Peruzzo, Jarrod McClean, Peter Shadbolt, Man-Hong Yung, Xiao-Qi Zhou,
  Peter~J Love, Al{\'a}n Aspuru-Guzik, and Jeremy~L O’brien.
\newblock A variational eigenvalue solver on a photonic quantum processor.
\newblock {\em Nature Communications}, 5(1):4213, 2014.

\bibitem{wecker2015progress}
Dave Wecker, Matthew~B Hastings, and Matthias Troyer.
\newblock Progress towards practical quantum variational algorithms.
\newblock {\em PRA}, 92(4):042303, 2015.

\bibitem{cerezo2021variational}
Marco Cerezo, Andrew Arrasmith, Ryan Babbush, Simon~C Benjamin, Suguru Endo,
  Keisuke Fujii, Jarrod~R McClean, Kosuke Mitarai, Xiao Yuan, Lukasz Cincio,
  et~al.
\newblock Variational quantum algorithms.
\newblock {\em Nature Reviews Physics}, 3(9):625--644, 2021.

\bibitem{larocca2024review}
Martin Larocca, Supanut Thanasilp, Samson Wang, Kunal Sharma, Jacob Biamonte,
  Patrick~J Coles, Lukasz Cincio, Jarrod~R McClean, Zo{\"e} Holmes, and
  M~Cerezo.
\newblock A review of barren plateaus in variational quantum computing.
\newblock {\em arXiv preprint arXiv:2405.00781}, 2024.

\bibitem{mcclean2017subspace}
Jarrod~R. McClean, Mollie~E. Kimchi-Schwartz, Jonathan Carter, and Wibe~A.
  de~Jong.
\newblock Hybrid quantum-classical hierarchy for mitigation of decoherence and
  determination of excited states.
\newblock {\em Phys. Rev. A}, 95:042308, Apr 2017.

\bibitem{parrish2019quantum}
Robert~M Parrish and Peter~L McMahon.
\newblock Quantum filter diagonalization: Quantum eigendecomposition without
  full quantum phase estimation.
\newblock {\em {arXiv preprint, arXiv:1909.08925}}, 2019.

\bibitem{motta2020qite_qlanczos}
Mario Motta, Chong Sun, Adrian T.~K. Tan, Matthew~J. O'Rourke, Erika Ye,
  Austin~J. Minnich, Fernando G. S.~L. Brand{\~a}o, and Garnet Kin-Lic Chan.
\newblock Determining eigenstates and thermal states on a quantum computer
  using quantum imaginary time evolution.
\newblock {\em Nature Physics}, 16(2):205--210, 2020.

\bibitem{klymko2022realtime}
Katherine Klymko, Carlos Mejuto-Zaera, Stephen~J. Cotton, Filip Wudarski,
  Miroslav Urbanek, Diptarka Hait, Martin Head-Gordon, K.~Birgitta Whaley,
  Jonathan Moussa, Nathan Wiebe, Wibe~A. de~Jong, and Norm~M. Tubman.
\newblock Real-time evolution for ultracompact {H}amiltonian eigenstates on
  quantum hardware.
\newblock {\em PRX Quantum}, 3:020323, May 2022.

\bibitem{Epperly_2022}
Ethan~N. Epperly, Lin Lin, and Yuji Nakatsukasa.
\newblock A theory of quantum subspace diagonalization.
\newblock {\em SIAM Journal on Matrix Analysis and Applications},
  43(3):1263–1290, August 2022.

\bibitem{shen2023realtimekrylov}
Yizhi Shen, Katherine Klymko, James Sud, David~B. Williams-Young, Wibe A.~de
  Jong, and Norm~M. Tubman.
\newblock Real-{T}ime {K}rylov {T}heory for {Q}uantum {C}omputing {A}lgorithms.
\newblock {\em {Quantum}}, 7:1066, July 2023.

\bibitem{yang2023dualgse}
Bo~Yang, Nobuyuki Yoshioka, Hiroyuki Harada, Shigeo Hakkaku, Yuuki Tokunaga,
  Hideaki Hakoshima, Kaoru Yamamoto, and Suguru Endo.
\newblock Dual-{GSE}: Resource-efficient generalized quantum subspace
  expansion.
\newblock {\em {arXiv preprint, arXiv:2309.14171}}, 2023.

\bibitem{yang2023shadow}
Ruyu Yang, Tianren Wang, Bing-Nan Lu, Ying Li, and Xiaosi Xu.
\newblock Shadow-based quantum subspace algorithm for the nuclear shell model.
\newblock {\em {arXiv preprint, arXiv:2306.08885}}, 2023.

\bibitem{ohkura2023leveraging}
Yasuhiro Ohkura, Suguru Endo, Takahiko Satoh, Rodney~Van Meter, and Nobuyuki
  Yoshioka.
\newblock Leveraging hardware-control imperfections for error mitigation via
  generalized quantum subspace.
\newblock {\em {arXiv preprint, arXiv:2303.07660}}, 2023.

\bibitem{kanno2023quantum}
Keita Kanno, Masaya Kohda, Ryosuke Imai, Sho Koh, Kosuke Mitarai, Wataru
  Mizukami, and Yuya~O Nakagawa.
\newblock Quantum-selected configuration interaction: Classical diagonalization
  of {H}amiltonians in subspaces selected by quantum computers.
\newblock {\em {arXiv preprint, arXiv:2302.11320}}, 2023.

\bibitem{ibm2024chemistry}
Javier Robledo-Moreno, Mario Motta, Holger Haas, Ali Javadi-Abhari, Petar
  Jurcevic, William Kirby, Simon Martiel, Kunal Sharma, Sandeep Sharma,
  Tomonori Shirakawa, Iskandar Sitdikov, Rong-Yang Sun, Kevin~J. Sung, Maika
  Takita, Minh~C. Tran, Seiji Yunoki, and Antonio Mezzacapo.
\newblock Chemistry beyond the scale of exact diagonalization on a
  quantum-centric supercomputer.
\newblock {\em Science Advances}, 11(25):eadu9991, 2025.

\bibitem{yoshioka2024diagonalization}
Nobuyuki Yoshioka, Mirko Amico, William Kirby, Petar Jurcevic, Arkopal Dutt,
  Bryce Fuller, Shelly Garion, Holger Haas, Ikko Hamamura, Alexander Ivrii,
  Ritajit Majumdar, Zlatko Minev, Mario Motta, Bibek Pokharel, Pedro Rivero,
  Kunal Sharma, Christopher~J. Wood, Ali Javadi-Abhari, and Antonio Mezzacapo.
\newblock Diagonalization of large many-body {H}amiltonians on a quantum
  processor.
\newblock {\em {arXiv preprint, arXiv:2407.14431}}, 2024.

\bibitem{motta2023subspace}
Mario Motta, William Kirby, Ieva Liepuoniute, Kevin~J Sung, Jeffrey Cohn,
  Antonio Mezzacapo, Katherine Klymko, Nam Nguyen, Nobuyuki Yoshioka, and
  Julia~E Rice.
\newblock Subspace methods for electronic structure simulations on quantum
  computers.
\newblock {\em Electronic Structure}, 6(1):013001, 2024.

\bibitem{oumarou2025molecular}
Oumarou Oumarou, Pauline~J Ollitrault, Cristian~L Cortes, Maximilian Scheurer,
  Robert~M Parrish, and Christian Gogolin.
\newblock Molecular properties from quantum krylov subspace diagonalization.
\newblock {\em arXiv preprint arXiv:2501.05286}, 2025.

\bibitem{kim2023evidence}
Youngseok Kim, Andrew Eddins, Sajant Anand, Ken~Xuan Wei, Ewout Van Den~Berg,
  Sami Rosenblatt, Hasan Nayfeh, Yantao Wu, Michael Zaletel, Kristan Temme,
  et~al.
\newblock Evidence for the utility of quantum computing before fault tolerance.
\newblock {\em Nature}, 618(7965):500--505, 2023.

\bibitem{shinjo2024unveiling}
Kazuya Shinjo, Kazuhiro Seki, Tomonori Shirakawa, Rong-Yang Sun, and Seiji
  Yunoki.
\newblock Unveiling clean two-dimensional discrete time quasicrystals on a
  digital quantum computer.
\newblock {\em arXiv preprint arXiv:2403.16718}, 2024.

\bibitem{kaliakin2024supramolecular}
Danil Kaliakin, Akhil Shajan, Javier~Robledo Moreno, Zhen Li, Abhishek Mitra,
  Mario Motta, Caleb Johnson, Abdullah~Ash Saki, Susanta Das, Iskandar
  Sitdikov, et~al.
\newblock Accurate quantum-centric simulations of supramolecular interactions.
\newblock {\em arXiv preprint arXiv:2410.09209}, 2024.

\bibitem{barison2024ext-sqd}
Stefano Barison, Javier~Robledo Moreno, and Mario Motta.
\newblock Quantum-centric computation of molecular excited states with extended
  sample-based quantum diagonalization.
\newblock {\em arXiv preprint arXiv:2411.00468}, 2024.

\bibitem{liepuoniute2024triplet}
Ieva Liepuoniute, Kirstin~D Doney, Javier Robledo-Moreno, Joshua~A Job, Will~S
  Friend, and Gavin~O Jones.
\newblock Quantum-centric study of methylene singlet and triplet states.
\newblock {\em arXiv preprint arXiv:2411.04827}, 2024.

\bibitem{shajan2024SQD-DMET}
Akhil Shajan, Danil Kaliakin, Abhishek Mitra, Javier~Robledo Moreno, Zhen Li,
  Mario Motta, Caleb Johnson, Abdullah~Ash Saki, Susanta Das, Iskandar
  Sitdikov, et~al.
\newblock Towards quantum-centric simulations of extended molecules:
  sample-based quantum diagonalization enhanced with density matrix embedding
  theory.
\newblock {\em arXiv preprint arXiv:2411.09861}, 2024.

\bibitem{alexeev2023quantum}
Yuri Alexeev, Maximilian Amsler, Paul Baity, Marco~Antonio Barroca, Sanzio
  Bassini, Torey Battelle, Daan Camps, David Casanova, Frederic~T Chong,
  Charles Chung, et~al.
\newblock Quantum-centric supercomputing for materials science: A perspective
  on challenges and future directions.
\newblock {\em arXiv preprint arXiv:2312.09733}, 2023.

\bibitem{Anderson1961}
P.~W. Anderson.
\newblock Localized magnetic states in metals.
\newblock {\em Phys. Rev.}, 124:41--53, Oct 1961.

\bibitem{white1992DMRG1}
Steven~R. White.
\newblock Density matrix formulation for quantum renormalization groups.
\newblock {\em Phys. Rev. Lett.}, 69:2863--2866, Nov 1992.

\bibitem{white1993DMRG}
Steven~R. White.
\newblock Density-matrix algorithms for quantum renormalization groups.
\newblock {\em Phys. Rev. B}, 48:10345--10356, Oct 1993.

\bibitem{white2005DMRG}
Steven~R. White.
\newblock Density matrix renormalization group algorithms with a single center
  site.
\newblock {\em Phys. Rev. B}, 72:180403, Nov 2005.

\bibitem{wu2022disentanglinginteractingsystemsfermionic}
Ang-Kun Wu, Benedikt Kloss, Wladislaw Krinitsin, Matthew~T Fishman, JH~Pixley,
  and EM~Stoudenmire.
\newblock Disentangling interacting systems with fermionic gaussian circuits:
  Application to quantum impurity models.
\newblock {\em Physical Review B}, 111(3):035119, 2025.

\bibitem{Barzykin1998Anderson}
Victor Barzykin and Ian Affleck.
\newblock Screening cloud in the $k$-channel kondo model: Perturbative and
  large-$k$ results.
\newblock {\em Phys. Rev. B}, 57:432--448, Jan 1998.

\bibitem{Holtzner2009Anderson}
Andreas Holzner, Ian~P. McCulloch, Ulrich Schollw\"ock, Jan von Delft, and
  Fabian Heidrich-Meisner.
\newblock Kondo screening cloud in the single-impurity anderson model: A
  density matrix renormalization group study.
\newblock {\em Phys. Rev. B}, 80:205114, Nov 2009.

\bibitem{moreno2023orbitalOptims}
Javier~Robledo Moreno, Jeffrey Cohn, Dries Sels, and Mario Motta.
\newblock Enhancing the expressivity of variational neural, and
  hardware-efficient quantum states through orbital rotations.
\newblock {\em arXiv preprint arXiv:2302.11588}, 2023.

\bibitem{roos1980complete}
Bj{\"o}rn~O Roos, Peter~R Taylor, and Per~EM Sigbahn.
\newblock A complete active space scf method (casscf) using a density matrix
  formulated super-ci approach.
\newblock 48(2):157--173, 1980.

\bibitem{head1988optimization}
Martin Head-Gordon and John~A Pople.
\newblock Optimization of wave function and geometry in the finite basis
  hartree-fock method.
\newblock 92(11):3063--3069, 1988.

\bibitem{werner1985second}
Hans-Joachim Werner and Peter~J Knowles.
\newblock A second order multiconfiguration scf procedure with optimum
  convergence.
\newblock 82(11):5053--5063, 1985.

\bibitem{olsen2011casscf}
Jeppe Olsen.
\newblock The casscf method: A perspective and commentary.
\newblock 111(13):3267--3272, 2011.

\bibitem{malmqvist1990restricted}
Per~{\AA}ke Malmqvist, Alistair Rendell, and Bj{\"o}rn~O Roos.
\newblock The restricted active space self-consistent-field method, implemented
  with a split graph unitary group approach.
\newblock 94(14):5477--5482, 1990.

\bibitem{zgid_density_2008}
Dominika Zgid and Marcel Nooijen.
\newblock The density matrix renormalization group self-consistent field
  method: {Orbital} optimization with the density matrix renormalization group
  method in the active space.
\newblock 128(14):144116, April 2008.

\bibitem{ghosh_orbital_2008}
Debashree Ghosh, Johannes Hachmann, Takeshi Yanai, and Garnet Kin-Lic Chan.
\newblock Orbital optimization in the density matrix renormalization group,
  with applications to polyenes and $\beta$-carotene.
\newblock 128(14):144117, April 2008.

\bibitem{wouters_density_2014}
Sebastian Wouters and Dimitri Van~Neck.
\newblock The density matrix renormalization group for ab initio quantum
  chemistry.
\newblock 68(9):272, September 2014.

\bibitem{ma_second-order_2017}
Yingjin Ma, Stefan Knecht, Sebastian Keller, and Markus Reiher.
\newblock Second-{Order} {Self}-{Consistent}-{Field} {Density}-{Matrix}
  {Renormalization} {Group}.
\newblock 13(6):2533--2549, June 2017.

\bibitem{JordanWigner1928}
P.~Jordan and E.~Wigner.
\newblock {\"U}ber das paulische {\"a}quivalenzverbot.
\newblock {\em Zeit. Phys}, 47(9):631--651, Sep 1928.

\bibitem{varbench}
Dian Wu, Riccardo Rossi, Filippo Vicentini, Nikita Astrakhantsev, Federico
  Becca, Xiaodong Cao, Juan Carrasquilla, Francesco Ferrari, Antoine Georges,
  Mohamed Hibat-Allah, Masatoshi Imada, Andreas~M. Läuchli, Guglielmo Mazzola,
  Antonio Mezzacapo, Andrew Millis, Javier~Robledo Moreno, Titus Neupert,
  Yusuke Nomura, Jannes Nys, Olivier Parcollet, Rico Pohle, Imelda Romero,
  Michael Schmid, J.~Maxwell Silvester, Sandro Sorella, Luca~F. Tocchio, Lei
  Wang, Steven~R. White, Alexander Wietek, Qi~Yang, Yiqi Yang, Shiwei Zhang,
  and Giuseppe Carleo.
\newblock Variational benchmarks for quantum many-body problems.
\newblock {\em Science}, 386(6719):296--301, 2024.

\bibitem{sugisaki2024SKQD}
Kenji Sugisaki, Shu Kanno, Toshinari Itoko, Rei Sakuma, and Naoki Yamamoto.
\newblock Hamiltonian simulation-based quantum-selected configuration
  interaction for large-scale electronic structure calculations with a quantum
  computer.
\newblock {\em arXiv preprint arXiv:2412.07218}, 2024.

\bibitem{mikkelsen2024SKQD}
Mathias Mikkelsen and Yuya~O Nakagawa.
\newblock Quantum-selected configuration interaction with time-evolved state.
\newblock {\em arXiv preprint arXiv:2412.13839}, 2024.

\bibitem{ffsim}
{The ffsim developers}.
\newblock {ffsim: Faster simulations of fermionic quantum circuits.}

\bibitem{sqd_addon}
Caleb Johnson, Stefano Barison, Bryce Fuller, James~R. Garrison, Jennifer~R.
  Glick, Abdullah~Ash Saki, Antonio Mezzacapo, Javier Robledo-Moreno, Max
  Rossmannek, Paul Schweigert, Iskandar Sitdikov, and Kevin~J. Sung.
\newblock {Qiskit addon: sample-based quantum diagonalization}.
\newblock \url{https://github.com/Qiskit/qiskit-addon-sqd}, 2024.

\bibitem{qiskit2024}
Ali Javadi-Abhari, Matthew Treinish, Kevin Krsulich, Christopher~J. Wood, Jake
  Lishman, Julien Gacon, Simon Martiel, Paul~D. Nation, Lev~S. Bishop,
  Andrew~W. Cross, Blake~R. Johnson, and Jay~M. Gambetta.
\newblock Quantum computing with {Q}iskit, 2024.

\bibitem{block2}
Huanchen Zhai, Henrik~R. Larsson, Seunghoon Lee, Zhi-Hao Cui, Tianyu Zhu, Chong
  Sun, Linqing Peng, Ruojing Peng, Ke~Liao, Johannes Tölle, Junjie Yang,
  Shuoxue Li, and Garnet Kin-Lic Chan.
\newblock Block2: A comprehensive open source framework to develop and apply
  state-of-the-art dmrg algorithms in electronic structure and beyond.
\newblock {\em The Journal of Chemical Physics}, 159(23):234801, 12 2023.

\bibitem{sun2018pyscf}
Qiming Sun, Timothy~C Berkelbach, Nick~S Blunt, George~H Booth, Sheng Guo,
  Zhendong Li, Junzi Liu, James~D McClain, Elvira~R Sayfutyarova, Sandeep
  Sharma, et~al.
\newblock {PySCF: the Python-based simulations of chemistry framework}.
\newblock {\em WIREs Comput. Mol. Sci}, 8(1):e1340, 2018.

\bibitem{sun2020recent}
Qiming Sun, Xing Zhang, Samragni Banerjee, Peng Bao, Marc Barbry, Nick~S Blunt,
  Nikolay~A Bogdanov, George~H Booth, Jia Chen, Zhi-Hao Cui, et~al.
\newblock Recent developments in the {PySCF} program package.
\newblock {\em J. Chem. Phys}, 153(2):024109, 2020.

\bibitem{kirby2024analysis}
William Kirby.
\newblock Analysis of quantum {K}rylov algorithms with errors.
\newblock {\em {Quantum}}, 8:1457, August 2024.

\bibitem{bravyi2024classical}
Sergey Bravyi, David Gosset, and Yinchen Liu.
\newblock Classical simulation of peaked shallow quantum circuits.
\newblock In {\em Proceedings of the 56th Annual ACM Symposium on Theory of
  Computing}, pages 561--572, 2024.

\bibitem{aaronson2024verifiable}
Scott Aaronson and Yuxuan Zhang.
\newblock On verifiable quantum advantage with peaked circuit sampling.
\newblock {\em {arXiv preprint, arXiv:2404.14493}}, 2024.

\bibitem{Borda2007Anderson}
L\'aszl\'o Borda.
\newblock Kondo screening cloud in a one-dimensional wire: Numerical
  renormalization group study.
\newblock {\em Phys. Rev. B}, 75:041307, Jan 2007.

\bibitem{zhai2023block2}
Huanchen Zhai, Henrik~R Larsson, Seunghoon Lee, Zhi-Hao Cui, Tianyu Zhu, Chong
  Sun, Linqing Peng, Ruojing Peng, Ke~Liao, Johannes T{\"o}lle, et~al.
\newblock Block2: A comprehensive open source framework to develop and apply
  state-of-the-art dmrg algorithms in electronic structure and beyond.
\newblock {\em The Journal of Chemical Physics}, 159(23), 2023.

\bibitem{sharma2017semistochastic}
Sandeep Sharma, Adam~A Holmes, Guillaume Jeanmairet, Ali Alavi, and Cyrus~J
  Umrigar.
\newblock Semistochastic heat-bath configuration interaction method: Selected
  configuration interaction with semistochastic perturbation theory.
\newblock {\em Journal of chemical theory and computation}, 13(4):1595--1604,
  2017.

\bibitem{wecker_solving_2015}
Dave Wecker, Matthew~B. Hastings, Nathan Wiebe, Bryan~K. Clark, Chetan Nayak,
  and Matthias Troyer.
\newblock Solving strongly correlated electron models on a quantum computer.
\newblock {\em Phys. Rev. A}, 92:062318, Dec 2015.

\bibitem{Kivlichan2018GivensSlater}
Ian~D. Kivlichan, Jarrod McClean, Nathan Wiebe, Craig Gidney, Al\'an
  Aspuru-Guzik, Garnet Kin-Lic Chan, and Ryan Babbush.
\newblock Quantum simulation of electronic structure with linear depth and
  connectivity.
\newblock {\em Phys. Rev. Lett.}, 120:110501, Mar 2018.

\bibitem{jiang_correlated_2018}
Zhang Jiang, Kevin~J. Sung, Kostyantyn Kechedzhi, Vadim~N. Smelyanskiy, and
  Sergio Boixo.
\newblock Quantum algorithms to simulate many-body physics of correlated
  fermions.
\newblock {\em Phys. Rev. Applied}, 9:044036, Apr 2018.

\bibitem{kingma2014Adam}
Diederik~P Kingma.
\newblock Adam: A method for stochastic optimization.
\newblock {\em arXiv preprint arXiv:1412.6980}, 2014.

\bibitem{jax2018github}
James Bradbury, Roy Frostig, Peter Hawkins, Matthew~James Johnson, Chris Leary,
  Dougal Maclaurin, George Necula, Adam Paszke, Jake Vander{P}las, Skye
  Wanderman-{M}ilne, and Qiao Zhang.
\newblock {JAX}: composable transformations of {P}ython+{N}um{P}y programs,
  2018.

\bibitem{pfeuty1970the-one-dimensi}
Pierre Pfeuty.
\newblock The one-dimensional {I}sing model with a transverse field.
\newblock {\em Annals of Physics}, 57(1):79--90, 1970.

\end{thebibliography}

\end{document}